\documentclass[11pt]{article}

\usepackage{fullpage,times}
\usepackage{cite}
\usepackage{algorithmic}
\usepackage{graphicx}
\usepackage{textcomp}
\usepackage{mathtools}
\usepackage{amsmath,amsthm,amssymb,amsfonts,MnSymbol}
\usepackage{mathrsfs}
\usepackage{color}
\usepackage[backref]{hyperref}
\usepackage{enumitem}
\usepackage{booktabs}
\usepackage{lineno}

\hypersetup{colorlinks=true, linkcolor=red,citecolor=blue,urlcolor=blue}

\makeatother

\usepackage{tikz}
\usetikzlibrary{arrows,chains,matrix,positioning,scopes}
\def\BibTeX{{\rm B\kern-.05em{\sc i\kern-.025em b}\kern-.08em
    T\kern-.1667em\lower.7ex\hbox{E}\kern-.125emX}}

\makeatletter
\tikzset{join/.code=\tikzset{after node path={%
\ifx\tikzchainprevious\pgfutil@empty\else(\tikzchainprevious)%
edge[every join]#1(\tikzchaincurrent)\varphii}}}
\makeatother

\tikzset{>=stealth',every on chain/.append style={join},
         every join/.style={->}}
\tikzstyle{labeled}=[execute at begin node=$\scriptstyle,
   execute at end node=$]

\newcommand{\hm}{\mathrm {hm}}

\providecommand{\dotdiv}{% Don't redefine it if available
  \mathbin{% We want a binary operation
    \vphantom{+}% The same height as a plus or minus
    \text{% Change size in sub/superscripts
      \mathsurround=0pt % To be on the safe side
      \ooalign{% Superimpose the two symbols
        \noalign{\kern-.35ex}% but the dot is raised a bit
        \hidewidth$\smash{\cdot}$\hidewidth\cr % Dot
        \noalign{\kern.35ex}% Backup for vertical alignment
        $-$\cr % Minus
      }%
    }%
  }%
}

\newcommand{\commentout}[1]{}

\newcommand{\supp}{\mathrm {supp}}

\newcommand{\butnot}{\mathrel{%
  \hspace{.1ex}
  \begin{tikzpicture}[baseline=-.57ex, line width=.125ex]
    \draw[-] (.2ex,0) --(1.5ex,0);
    \draw[-, line width=.01ex, fill=black]
             (1.35ex,0) -- (1.84ex, .48ex)
                       -- (1.91ex ,.418ex)
                       -- (1.55ex,   0ex)
                       -- (1.91ex ,-.418ex)
                       -- (1.84ex,-.48ex)
                       -- (1.35ex,0ex);
  \end{tikzpicture}
\hspace{.1ex}}}

\newtheorem{Def}{Definition}
\newtheorem{Thm}{Theorem}
\newtheorem{Rmk}{Remark}
\newtheorem{Exm}{Example}

\newtheorem{Lem}{Lemma}
\newtheorem{Pro}{Proposition}

\usepackage{color,dsfont}

\theoremstyle{plain}

\theoremstyle{definition}
\usepackage[]{graphicx}

\title{Codd's Theorem for Databases over Semirings} 

\author{Guillermo Badia \footnote{School of Historical and Philosophical Inquiry, University of Queensland, Brisbane, St Lucia, QLD 4072, Australia.}
\and
Phokion G. Kolaitis \footnote{
University of California Santa Cruz, USA and IBM  Research,  USA.}
\and
Carles Noguera \footnote{Department of Information Engineering and Mathematics,  University of Siena 
         San Niccol\`o, via Roma 56, 53100 Siena, Italy.}}

\begin{document}

\maketitle

%TODO mandatory: add short abstract of the document
\begin{abstract}
Codd's Theorem, a fundamental result of database theory,  asserts that relational algebra and relational calculus have the same expressive power on relational databases. We explore Codd's Theorem for databases over semirings and establish two different versions of this result for such databases: the first version involves the five basic operations of relational algebra, while in the second version the division operation is added to the  five basic operations of relational algebra. 
In both versions, the difference operation of relations is given semantics using semirings with monus, while on the side of relational calculus a limited form of negation is used. 
The reason for considering these two different versions of Codd's theorem is that, unlike the case of ordinary relational databases, the division operation need  not be expressible in terms of the five basic operations of relational algebra for databases over an arbitrary positive semiring; in fact, we show that this inexpressibility result holds even for bag databases.  
\end{abstract}

\section{Introduction}

% \guillermo{ TO DO: (1) Define the semantics of formulas precisely (pay attention to what happens to tuples in conjunctive formulas), (2) Check if having no zero divisors is equivalent to having Codd's theorem for the language including universal quantifier and division for semirings with monus and zero sum free (if so, give example of a semiring where we would have Codd's theorem 1 but not 2; consider the power set semiring whihc is zero sum free but with many zero divisors and look at the diamond lattice four valued logic), (3) make a note at the end regarding the duality problem with semirings (it's difficult because of not assuming idemptotence)  }

\noindent{\bf Background and Motivation}~ 
Codd's Theorem~\cite{Codd}
is 
%the first non-trivial and 
arguably the most fundamental result in relational database theory. Informally, Codd's Theorem asserts that relational algebra and relational calculus have the same expressive power on relational databases. The significance of Codd's Theorem is that it demonstrated that  a procedural query language based on operations on relations and a declarative query language based on first-order logic can express precisely the same queries. As regards its impact, Codd's Theorem gave rise to the notion of \emph{relational completeness} as a benchmark for the expressive power of query languages. Furthermore, it became the catalyst for numerous subsequent investigations of both sublanguages of relational algebra and relational calculus, such as conjunctive queries, and extensions of relational
algebra and relational calculus, such as Datalog.

In more recent years, there has been an extensive study of databases over semirings, i.e., databases in which the tuples of the relations are annotated with elements from a fixed semiring. Standard relational databases are, of course, databases over the Boolean semiring $\mathbb{B}=(\{0,1\}, \vee, \wedge, 0,1)$, while \emph{bag} databases are databases over the semiring $\mathbb{N}=(N, +, \cdot, 0,1)$ of the natural numbers. The interest in semirings other than the Boolean  and the bag ones was sparked by the influential paper~\cite{GKV}
in which semirings of polynomials were used to 
formalize and study the \emph{provenance} of database queries. 
Originally, the semiring-based study of provenance focused on the positive fragment of relational algebra (i.e., unions of conjunctive queries), but it was subsequently extended to the study of provenance in first-order logic~\cite{G-T} and least fixed-point logic~\cite{DBLP:conf/csl/DannertGNT21}.
Researchers  also investigated a variety of other database problems under semiring semantics, including the query containment problem~\cite{Green:2011,DBLP:journals/tods/KostylevRS14}, the evaluation of Datalog queries~\cite{DBLP:journals/jacm/KhamisNPSW24,DBLP:journals/pacmmod/ZhaoDKRT24}, 
and the interplay between local   consistency and global consistency for relations over  semirings~\cite{AtseriasK2023,DBLP:journals/pacmmod/AtseriasK24}. In a different direction, connections between databases over semirings and quantum information theory were established~\cite{DBLP:conf/birthday/Abramsky13}.

In view of the aforementioned developments, it is rather surprising that the following questions have not been addressed thus far: Is there a Codd's Theorem for databases over semirings? In other words, do relational algebra and relational calculus have the same expressive power on databases over semirings? In particular, is there a version of Codd's Theorem for bag databases?
Our goal in this paper is to investigate these questions and to provide some answers.

 Before proceeding further, we need to take a closer look at the precise statement of Codd's Theorem. In a literal sense, it is simply \emph{not} true that relational algebra and relational calculus have the same expressive power over relational databases. One of reasons for this is that relational calculus allows for the use of unrestricted negation, which implies that relational calculus formulas can define infinite relations. Codd~\cite{Codd} was aware of this issue, thus he carefully crafted a collection of relational calculus formulas, called \emph{$\alpha$-expressions}, and showed that $\alpha$-expressions have the same expressive power as relational algebra. Since the syntax of $\alpha$-expressions is somewhat convoluted, there was extensive follow-up research on richer syntactically-defined fragments of relational calculus that have the same expressive power as relational algebra (see, e.g., \cite{DBLP:journals/tods/GelderT91}). In a parallel direction, it was also realized that the equivalence between relational algebra and relational calculus holds if the relational calculus formulas are \emph{domain independent}, which means that they return the same answers as long as the universe of discourse contains the \emph{active domain} of the database at hand. This version of
Codd's Theorem is the one presented in standard database theory textbooks (e.g., in~\cite{AbiVi}). In what follows, we will explore whether and how this version of Codd's Theorem extends to databases over semirings.

\smallskip

\noindent{\bf Analysis of Choices}~
A semiring is a structure $\mathbb{K}=(K,+,\cdot,0,1)$, where addition $+$ and multiplication $\cdot$ are binary operations on $K$, and certain identities hold (precise definitions are given in the next section). When  positive relational algebra is interpreted  over semirings, $+$ and $\cdot$ are used to give semantics to union $\cup$  and Cartesian product $\times$, respectively.  Similarly, $+$ and $\cdot$ are used to give semantics to disjunction $\vee$  and conjunction $\wedge$, when unions of conjunctive queries are interpreted over semirings. However, relational algebra also contains the difference operation $\setminus$ as a basic operation.
Therefore, to interpret full relational algebra over semirings,  we must expand the semirings  with a ``difference'' operation. It is well understood, though, that there is no unique way to expand a semiring with a well-behaved ``difference'' operation (see~\cite{DBLP:conf/birthday/Suciu24} for a thorough coverage of this issue). Thus, a choice has to be made concerning the meaning of the ``difference'' operation on semirings.

Here, we opt to focus on semirings with a \emph{monus} operation, an operation that is well defined for \emph{naturally ordered} semirings with some additional properties. 
In the study of provenance of arbitrary relational algebra expressions, 
the  monus operation was adopted  in~\cite{GP} as an interpretation of the difference operation on relations;  this choice was  criticized in~\cite{DBLP:conf/tapp/AmsterdamerDT11} on the grounds that certain desirable algebraic identities involving natural join and difference fail on some semirings with monus.  Our decision to focus on semirings with monus is based on two considerations: first, semirings with monus comprise a large class of semirings that, beyond the Boolean semiring, includes the bag semiring, the tropical semiring on the natural numbers, and 
every  complete bounded distributive lattice (hence, the fuzzy semiring);
 second, we do not believe that the failure of some useful identities on some semirings with monus is sufficient justification to disregard semirings with monus. After all, such fundamental identities as $R\cup R = R$ fail on the bag semiring, yet bag semantics is the default SQL semantics. And, of course, if one wishes to establish a version of Codd's Theorem for bag databases, then the monus operation on bags has to be used, since, this is  what is used in the semantics of SQL.

On the side of relational calculus, we have to make a  choice concerning negation.
In the study of provenance of first-order formulas,  Gr\"adel and Tannen~\cite{G-T} allowed for   unrestricted negation, but the formulas are always assumed to be in negation normal form (i.e., the negation symbol is pushed in front of atomic formulas). 
Furthermore, the notion of an
\emph{interpretation} is used to give semiring semantics to 
 first-order formulas in negation normal form, where an interpretation
is an assignment of semiring values to every atomic and negated atomic statements. Note also that Gr\"adel and Tannen~\cite{G-T} do not explicitly define what it means for a database over a semiring $\mathbb K$ to satisfy a first-order formula;  such a definition, however, could be obtained by fixing a particular \emph{canonical}
interpretation associated with the database of interest. We believe that the 
restriction to formulas in negation normal form is rather severe from a query language point of view. 
For this reason, we do not require the formulas to be in any normal form, but we also avoid unrestricted negation and, instead, we opt to use the binary connective $\butnot$ (``but not'').  This connective expresses a limited form of negation and has a long history of uses in  algebra and in non-classical logics~\cite{MT,Good,rauszer1974semi}. %In Section 
%\ref{connections:sec}, we elaborate on  the uses of $\butnot$ in  the dual intuitionistic logic and the dual G\"odel-Dummet logic.
We note that, when logics are given semiring semantics, 
 the use of the connective $\butnot$  to model negation was criticized in~\cite{G-T3}. That criticism was mainly based on the argument that the connective $\butnot$ does not give rise to a ``reasonable'' form of negation, when the negation $\neg \varphi$ of a formula $\varphi$  is defined as $1\butnot \varphi$, i.e., when $\butnot$ is used to define an unrestricted form of negation. As mentioned earlier, however, the connective $\butnot$ has been traditionally used to model a restricted form of negation (a ``safe'' negation, if you will) and not an unrestricted negation. In our setting,  we need a restricted form of negation that is the counterpart of the monus operation.

%We see two main sources of dissatisfaction with the semantics in~\cite{G-T, G-T2}, even though it has the advantage of great generality (it works for any commutative semiring). The first is the downside that it requires a foreign element that does not come in practice with semiring-annotated databases: so called \emph{interpretations}, which are external logical tools used to interpret negation. The second issue is that formulas are always assumed to be given in negation normal form, which makes perfect sense in the Boolean case, but becomes completely unnatural in non-classical contexts where formulas are also given values over semirings, such as intuitionistic logic (since, in general, there is no negation normal form result).

There are two other important issues to bring up  at this point. First, when one starts thinking about a possible Codd's Theorem for databases over semirings, one quickly realizes that relational algebra has to be expanded with a \emph{support} operation. This operation takes as input a relation over a semiring, removes all annotations from the semiring, and returns the underlying (ordinary) relation as output. For the Boolean semiring, there is no need to consider the support as an additional operation, since the support of a relation is the relation itself. For relations over the fuzzy semiring, however, it is easy to show that the support is not expressible in terms of the basic operations of relational algebra, thus the support operation has to be added to the syntax of relational algebra for relations over semirings. In turn, this implies that we need to consider semirings expanded with a support operation $\mathsf{s}$ such that
$\mathsf{s}(a)=1$ if $a\neq 0$, while $\mathsf{s}(a) =0$ if $a=0$.

The second issue is more subtle and, as it turns out, more consequential. Early on, Codd showed that the five basic operations of relational algebra (union, difference, Cartesian product, selection, and projection) can express the \emph{division} (or the  \emph{quotient}) of two relations~\cite{Codd}.  Division is a powerful operation that can be used to express universal quantification in the translation of relational calculus to relational algebra.
The division operation has natural semantics for relations over semirings.
At first, one may expect that division is always expressible in terms of the five basic relational algebra operations when we consider relations over semirings. Unfortunately, this is not true. Using a combinatorial argument, we show that already on the bag semiring 
$\mathbb{N}=(N,+, \cdot,0,1)$, the division operation is \emph{not} expressible in terms of the five basic relational algebra operations. Thus, to obtain a version of Codd's Theorem for relational calculus formulas involving universal quantification, we need to expand relational algebra with a division operation.

\smallskip

\noindent{\bf Contributions} With the insights gained from the preceding analysis, we proceed to establish the following two versions of Codd's Theorem for databases over semirings, one for relational algebra without division and one for relational algebra with division. 

Let $\mathbb{K}=(K,+,\cdot,-,\mathsf{s},0,1)$ be a 
zero-sum-free commutative semiring  with monus  and support, where \emph{zero-sum-free} means that if $a+b=0$, then $a=0$ and $b=0$.  We write $\mathcal{BRA}$ to denote the collection of all \emph{basic relational algebra} expressions, that is, all relational algebra expressions over a schema built using the operations union, difference, Cartesian product, selection, projection, and support. We also write $\mathcal{BRC}$ to denote the collection of all \emph{basic relational calculus} formulas, that is, all formulas built from atomic formulas, using the connectives $\wedge$, $\vee$, $\butnot$, and existential quantification. We show that a query $q$ on $\mathbb{K}$-databases is expressible by a $\mathcal{BRA}$-expression if and only if it is definable by a domain independent
$\mathcal{BRC}$-formula (Theorem~\ref{Codd1}).

Let $\mathbb{K}=(K,+,\cdot,-,\mathsf{s},0,1)$ be a 
positive semiring with monus and support, where \emph{positive} means that $\mathbb{K}$ is commutative, zero-sum-free, and has no zero-divisors (i.e., if $a\cdot b=0$, then $a=0$ or $b=0$).  We write $\mathcal{RA}$ to denote the collection of all \emph{relational algebra} expressions, that is,  all
 expressions over a schema built using the operations union, difference, Cartesian product, selection, projection, support, and division. We also write $\mathcal{RC}$ to denote the collection of all \emph{relational calculus} formulas, that is,  all formulas built from atomic formulas, using the connectives $\wedge$, $\vee$, $\butnot$, and both existential and universal  quantification. 
 We   show that a query $q$ on $\mathbb{K}$-databases is expressible by an $\mathcal{RA}$-expression if and only if it is definable by a domain independent
$\mathcal{RC}$-formula (Theorem~\ref{Codd2}). An immediate corollary of this result is a version of Codd's Theorem for bag databases.  In our view, this corollary alone justifies the choices of the monus operation and of the connective
$\butnot$ that we made here.

The proofs of these two results proceed by induction on the construction of the (basic) relational algebra expressions and the construction of the (basic) relational calculus formulas. 
%Once the building blocks of these formalisms are identified and carefully spelled out, the proofs go through with no surprises.
%are quite straightforward. 
Some of the inductive steps, require a rather delicate analysis because of the nuances  that arise when $\mathbb K$-databases over  arbitrary positive semirings are considered.
Note also that in the simulation of (basic) relational calculus by
 (basic) relational algebra, the proofs make repeated use of the fact that there is a relational algebra expression that uniformly defines the active domain of a $\mathbb{K}$-database, a fact in which the use of the support operation is of the essence.

 In conclusion, our work makes a  contribution to the ongoing study of databases over semirings by investigating Codd's Theorem in this context.  Specifically, our work identifies the ``right'' constructs of relational algebra and relational calculus  needed to establish Codd's Theorem for databases over semirings,  delineates the features of these constructs, unveils differences between the Boolean semiring and other semirings, and potentially  establishes benchmarks for comparing the expressive power of these and other formalisms under semiring semantics.
 
% In contrast to the situation for set semantics, as far as we are aware, there is no  Codd theorem along the lines of the classical result from~\cite{Codd} linking these two approaches and showing that both languages are expressively equivalent with suitable restrictions to first-order logic. We take such a result to be a desirable  benchmark for a natural semantics  for first-order logic in a semiring context.

%\medskip
%\noindent
%{\bf Our contribution.} The present paper develops a new conceptual framework for the logical treatment of negation (that falls in line with the treatment of difference in relational algebra from~\cite{GP}) in first-order logic  over semiring-annotated databases which deviates from ~\cite{G-T, G-T2}. We prove a semiring generalization of Codd's theorem in order to test our approach. One advantage of our proposed setting is that it fits well within the context of well-known logical languages  in the field of non-classical logics.

\section{Semirings with Monus and Support}

\begin{Def}[Semiring] A \emph{semiring} is a  structure $\mathbb{K}=(K, +, \cdot,0,1)$, where  addition $+$ and multiplication $\cdot$ are binary operations on $K$, and  $0$, $1$ are elements of $K$ such that the following properties hold: 
\begin{itemize}
\item $(K, +, 0)$  is a commutative monoid,  $(K, \cdot, 1)$ is a monoid, and $0\neq 1$;
\item (distributivity): for all elements $a,b,c$ in $K$, we have that 
$a\cdot (b+c) = a\cdot b + a \cdot c$;
\item (annihilation): for every element $a$ in $K$, we have that $a \cdot 0 = 0 \cdot a = 0$. 
\end{itemize}
\end{Def}

\noindent Let $\mathbb{K}=(K, +, \cdot,0,1)$ be a semiring. We will consider several  additional properties.
\begin{itemize}

\item  $\mathbb{K}$ is \emph{commutative}  if the monoid $(K,\cdot,1)$ is commutative.

\item $\mathbb{K}$ is \emph{zero-sum-free} if for all elements $a, b \in K$ with $a+b=0$, we have that $a=0$ and $b=0$.
\item 
$\mathbb{K}$ has
\emph{no zero divisors} if for all elements $a, b \in K$ with  $a\cdot b=0$, we have that $a=0$ or $b=0$.
\item  $\mathbb{K}$ is a  \emph{positive} semiring if it is commutative,  zero-sum-free, and has no zero divisors.
\end{itemize}

As is well known, a \emph{ring} is a semiring
$\mathbb{K}=(K, +, \cdot,0,1)$ such that
$(K,+,0)$ is a commutative group. In other words, a ring is a commutative semiring in which every element has an additive inverse, i.e., for every $b\in K$, there is an element $-b \in K$ such that
$b+(-b) = 0$. Therefore, on every ring, a binary \emph{difference} operation $-$ can be defined by setting $a-b = a+ (-b)$. This raises the question: is it  possible to expand a semiring with a binary operation that possesses some of the properties of the difference operation on rings. As discussed in~\cite{DBLP:conf/birthday/Suciu24}, there are several alternatives to defining such an operation. Here, we will focus on the \emph{monus} operation~\cite{amer1984equationally,GP}, which, however, requires that the semiring at hand is \emph{naturally ordered}. As we will see, not every naturally ordered semiring can be expanded to a semiring with monus.

If $\mathbb{K}=(K,+,\cdot,\mathsf{s},0,1)$ is a semiring, then the 
\emph{natural preorder} of $\mathbb{K}$ is the binary relation $\preceq$ on $K$ defined as follows: $a\preceq b$ if there is an element $c\in K$ such that $a+c=b$. Clearly,  $\preceq$ is a reflexive and transitive relation on $K$, i.e., $\preceq$ is indeed a preorder. If $\preceq$ happens to be a partial order (that is,  $a\preceq b$ and $b\preceq a$ imply  $a=b$), then  $\preceq$~is called the \emph{natural order} of the semiring $\mathbb{K}$. In this case, we say that $\mathbb{K}$ is a \emph{naturally ordered}
semiring.
 
\begin{Def}[Semiring with monus] \label{msemring} A  \emph{semiring with monus} is a structure  $\mathbb{K}= (K, +, \cdot, -,  0, 1)$ such that the following  hold:
\begin{itemize}
\item $(K, +, \cdot, 0, 1)$ is a naturally ordered commutative semiring;
\item For all elements $a,b \in K$,  there is a smallest (in the sense of the natural order $\preceq$) element $c$ in $K$, denoted by $a-b$,  such that $a\preceq b+c$.
\end{itemize}
\end{Def}

The next proposition provides different characterizations of the monus operation (see~\cite{amer1984equationally,DBLP:conf/tapp/AmsterdamerDT11}).
\begin{Pro}\label{monusprop}
Let\/ $\mathbb{K}= (K, +, \cdot,  0, 1)$ be a naturally ordered semiring and let $-'$ be a binary operation on $K$. Then, the following statements are equivalent:
\begin{itemize}
\item For all $a,b \in K$, we have $a -' b = a-b$, i.e., 
$a -' b$ is the smallest element c such that
$a\preceq b+c$.

\item For all $a,b, c \in K$, we have: $a -' b \preceq c$~ if and only if~ $a\preceq b+c$.
\item  The following identities hold for  $-'$:
\begin{enumerate}
\item $a -' a = 0$
\item $0 -' a = 0$
\item  $a + (b -' a) = b + (a -' b)$
\item  $a -' (b + c) = (a -' b) -' c$.
%\item[(v)] $a \cdot (b - c) = a \cdot b - a \cdot c$
\end{enumerate}
\end{itemize}
\end{Pro}
% We will further assume that the semirings we are working with are \emph{positive}:
% \begin{itemize}
% \item $0\neq 1$,
% \item  for all $a, b \in K$, if $ab=0$, then $a=0$ or $b=0$,
% \item  for all $a, b \in K$, if $a+b=0$, then $a=0$ or $b=0$.
% \end{itemize}

We now present several examples of important semirings with monus.

\begin{Exm}\label{bool} The two-element Boolean semiring $\mathbb{B}= (\{\top, \bot\}, \wedge, \vee, \bot, \top)$ can be expanded to a semiring with monus in which $a -b:=a\wedge b^*$, where $*$ is the standard Boolean complement.
\end{Exm}

\begin{Exm}
The semiring $\mathbb{N}=  (N, +, \cdot,  0, 1)$ of the natural numbers, also known as the \emph{bag} semiring, is a naturally ordered semiring in which the natural order coincides with the standard order $\leq$ on $N$. The bag semiring 
can be expanded to a semiring with monus in which
the monus operation is the \emph{truncated subtraction}  $a \dotdiv b$, that is,  $a \dotdiv b= a -b$ if $a\geq b$, while $a \dotdiv b =0$ if $a<b$. 
\end{Exm}

\begin{Exm}\label{latt}
Every bounded distributive lattice $\mathbb{D
}=(L,\vee,\wedge,\bot,\top)$ is a naturally ordered semiring in which $a\preceq b$ if and only if $a\vee b = b$. Such a lattice is \emph{complete} if every non-empty subset of $L$ has a smallest element. Every complete and bounded distributive lattice can be expanded to a semiring with monus in which $a-b = \min \{c: a \preceq b \vee c\}$.
%(see~\cite{MT}).
%The monus in these  semirings is known as \emph{pseudo-difference}. 
Complete bounded distributed lattices expanded with monus are called \emph{Brouwerian algebras}, a term coined by   McKinsey and Tarski~\cite{MT}. A prominent example of a Brouwerian algebra  is the \emph{fuzzy} semiring $\mathbb{F}=([0,1], \max, \min,0,1)$ expanded with monus; in this case, we have that $a-b = a$ if $a > b$, while
$a-b = 0$ if $a\leq b$.
\end{Exm}

\begin{Exm}\label{doubt} The \emph{\L ukasiewicz} semiring $\mathbb{L}=([0, 1], \max, \odot,0,1)$,
where $a \odot  b= \max(a+b-1,0)$, is a naturally ordered semiring in which the natural order coincides with the standard order on $[0,1]$.
 It can be  expanded to a semiring with
 monus in which $a-b = a$ if $a>b$, while  $a-b=0$ if $a\leq b$. 
\end{Exm}

%\begin{Exm}\label{doubt} The \emph{doubt} semiring $\mathbb{D}=([0, 1], \min, \oplus,1,0)$ (see~\cite{brinke_et_al:LIPIcs.CSL.2024.19}) is a naturally ordered semiring in which the natural order coincides with the standard order $\leq $ on $[0,1]$ and $a\oplus b = \text{min}(a+b, 1)$ (this is the dual of the \L ukasiewicz semiring $\mathbb{L}$, see~\cite{brinke_et_al:LIPIcs.CSL.2024.19}). This semiring can be easily expanded with a monus operator $-$ by letting  $a-b = \text{\emph{max}}\{0, a-b\}$. The fact that this is a monus follows from Proposition \ref{monusprop} and~\cite[Example 4]{10.1007/978-3-642-31715-6_54}. 
%\end{Exm}

\begin{Exm} The \emph{tropical} semiring $\mathbb{T}(R) = (\mathbb{R}\cup\{\infty\},\min, +, \infty, 0)$ of the real numbers is a semiring in which the natural order is the reverse standard order on $R$, i.e., $a\preceq b$ if and only if $a\geq b$.
It can be expanded with monus by letting
$a-b = a$ if $a < b$ (in the standard order of the reals), while $a-b =  \infty$ if $a\geq b$. Similarly,  the tropical semiring $\mathbb{T}(N) = (\mathbb{N}\cup\{\infty\},\min,+,\infty,0)$ of the natural numbers can be expanded to  a semiring with monus (namely, the restriction of the previously defined operation to this semiring).

\end{Exm}

\begin{Exm} \label{exam:suciu} Let $\mathbb{N}[X]$ be the  semiring of all polynomials with variables from  a set $X=\{x_1, \dots, x_n\}$ and with coefficients from the set $N$ of non-negative integers. This semiring is called the \emph{provenance} semiring because of its use in the study of database provenance \cite{G-T}. As pointed out in~\cite[Example 9]{GP}, the provenance semiring $\mathbb{N}[X]$ can be expanded to a semiring with monus in the following way. For every tuple $\alpha=(\alpha_1,\ldots,\alpha_n) \in N^n$, let $x^\alpha$
denote the monomial
 $x_1^{\alpha_1} x_2^{\alpha_2} \dots x_n^{\alpha_n}$, where we set $x_i^0=1$.
 For a finite set $I\subseteq \mathbb{N}^n$, let $f[X] = \sum_{\alpha \in I}f_\alpha x^\alpha$ and  $g[X] = \sum_{\alpha \in I}g_\alpha x^\alpha$ be  two polynomials in $\mathbb{N}[X]$. The monus is defined by letting $f[X] - g[X] = \sum_{\alpha \in I} (f_\alpha \dotdiv g_\alpha)x^\alpha$, where $\dotdiv$ is the truncated difference on $N$, that is, $f_\alpha\dotdiv g_\alpha= f_\alpha - g_\alpha$ if $f_\alpha \geq g_\alpha$, and $f_\alpha\dotdiv g_\alpha=0$ if $f_\alpha < g_\alpha$.
\end{Exm}

It should be noted that not every naturally ordered semiring can be expanded to a semiring with monus.  An example of such a semiring can be found in \cite{Monet2016}. Here we give a different example  that was provided to us by Dan Suciu \cite{Suciu2025}. 

\begin{Exm} Consider the linear order $(L,\leq)$, where $L=\{\mathrm I, \mathrm T, \mathrm S, \mathrm C, \mathrm P\}$ 
and $ \mathrm I < \mathrm T < \mathrm S < \mathrm C < \mathrm P$. Intuitively, the elements of this linear order are \emph{security} levels, where $\mathrm I$ stands for ``inaccessible", $\mathrm T$ stands for ``top secret", $\mathrm S$ stands for ``secret", $\mathrm C$ stands for ``confidential" and $\mathrm P$ stands for ``public".

Let $K  = \{(x, s) :  x \in \{ 1,2\ldots, \}, s \in L\} \cup \{0\}$.  Equip $K$ with two  operations $+$ and
$*$ defined as follows:
\begin{itemize}
\item $0 + k  =  k + 0  =  k$, for all $k \in K$;
\item $0*k  =  k*0  =  0$,  for all $k \in K$;
\item $(x,s) + (y,s') = (x+y, \min(s,s'))$, for all $x, x'\in \{1,2,\ldots\}$ and for all $s, s' \in L$;
\item $(x,s) * (y,s') = (x\times y, \min(s,s'))$, 
for all $x, x'\in \{1,2,\ldots\}$ and for all $s, s' \in L$.
\end{itemize}
In words,
the operations $+$ and $*$ are the standard addition and multiplication on $\{1, 2, \ldots\}$, while they keep the highest security level on $L$.  
%Notice that 0 does not have a security level.

Consider the structure  ${\mathbb K}=(K, +, *, 0, (1,\mathrm P))$. It is easy to verify that $\mathbb K$ is a  commutative semiring with $0$ as the neutral element of $+$ and $(1,P)$ as the neutral element of $*$. 
Furthermore, $\mathbb K$ is naturally ordered, where the natural order $\preceq$ is as follows:
\begin{itemize}
\item $0$ is the smallest element;
\item $(x,s) \preceq (x',s')$ if and only if
$x \leq  x'$ and $s\geq s'$.
\end{itemize}
  For example, $(42, \mathrm T) \prec (43,\mathrm T)$ and $(42,\mathrm T) \prec (43,\mathrm I)$. In contrast,  $(42,\mathrm T)$ and $(42,\mathrm I)$ are incomparable, because there is no element $(x,s)\in K$ such that
  $(42,\mathrm I)+(x,s) = (42,\mathrm T)$
  or $(42,\mathrm T)+(x,s) = (42,\mathrm I)$.
  
 We now claim that $\mathbb K$ cannot be expanded to a semiring with monus. Towards a contradiction, assume that $\minus$ is a monus operation of $\mathbb K$. Then
 $(43,\mathrm I)\minus (1,\mathrm I)$ should be the smallest element of  the set 
 $$P= \{(x,s)\in K: (43,\mathrm I)\preceq (1,\mathrm  I)+(x,s)\}.$$
 However, $P$ does not have a smallest element because
 $$P= \{(42,\mathrm I), (42,\mathrm T),
 (42,\mathrm S), (42,\mathrm C), (42,\mathrm P),  (43,\mathrm I), (43,\mathrm T),(43,\mathrm S), (43,\mathrm C), (43,\mathrm P),  (44,\mathrm I), \ldots\}.$$
\end{Exm}

We now introduce the \emph{support} operation on semirings.

\begin{Def}[Semiring with support]
A \emph{semiring with support} is a structure
$\mathbb{K}=(K, +, \cdot,\mathsf{s},0,1)$  such that $(K, +, \cdot,0,1)$ is a semiring and $s$ is the unary operation on $K$ such that for every $a \in K$, 
$\mathsf{s} (a) =\begin{cases} 1 & \text{if }a \neq 0; \\ 0 & \text{otherwise.} \end{cases}$
\end{Def}

Unlike monus, every semiring can be expanded to a semiring with support (and in a unique way).

\begin{Def}
[Semiring with monus and  support]
A \emph{semiring with monus and support} is a structure $\mathbb{K}=(K, +, \cdot,-,\mathsf{s},0,1)$ such that $(K, +,\cdot,-,0,1)$ is a semiring with monus and $(K, +, \cdot,\mathsf{s},0,1)$ is a semiring with support.
\end{Def}

\section{Codd's Theorem for Relational Algebra with Monus and Support}

\subsection{Semiring-annotated Relations  and   Basic Relational Algebra}

In what follows, we assume that $D$ is a fixed denumerable set, which  will be the \emph{domain} of elements occurring in databases. For  $n\geq 1$, we write $D^n$ to denote the set of all $n$-tuples from $D$. We also write  $D^0$ to denote the singleton  $\{ <  >\}$, where $<  >$ is the empty tuple (we assume that $<  > \ \notin D$). 

\begin{Def} \label{rel-defn} Let  $\mathbb{K}=(K, +, \cdot,  0,1)$ be a semiring and let $n\geq 0$ be a non-negative integer.
\begin{itemize} 
\item An $n$-ary $\mathbb{K}$-relation $R$ is a function $R \colon D^n\longrightarrow K$ of finite support, that is to say, the set $\supp(R) := \{(d_1,\ldots,d_n)\in D^n: R(d_1,\ldots,d_n) \neq 0\}$ is finite.  We call $n$ the \emph{arity} of $R$.

\item A \emph{$\mathbb{K}$-relation} is an $n$-ary $\mathbb{K}$-relation $R$, for some $n\geq 0$.

\item If $n\geq 1$ and $A\subseteq D$, then an \emph{$n$-ary $\mathbb{K}$-relation on $A$} is an $n$-ary $\mathbb{K}$-relation $R$ with $\supp(R)\subseteq A^n$.

\end{itemize}
\end{Def}

Note that there are  $|K|$-many $0$-ary $\mathbb{K}$-relations. Indeed,  for every $a\in K$, we have the $0$-ary relation $R^0_a$ such that $R^0_a(<  >)= a$. If $a\neq 0$, then $\supp (R^0_a) = \{<  >\}$, while $\supp (R^0_0) = \emptyset$.

Next, we will define several different operations on $\mathbb{K}$-relations. Before doing so, we need to introduce some useful terminology and notation.
\begin{itemize}
\item 
Let ${\bf t}=(d_1,\ldots,d_n) \in D^n$ be an $n$-tuple of elements from $D$, where $n\geq 1$.  If $v=(i_1,\ldots,i_k)$ is a $k$-tuple of distinct indices from $\{1,\ldots, n\}$, then ${\bf t}[v]$ denotes the \emph{projection of ${\bf t}$ on $v$}, that is to say,
${\bf t}[v]=(d_{i_1},\ldots,d_{i_k}) \in D^k$. 
Observe that, if $k=n$, then ${\bf t}[v]$  is a permutation of ${\bf t}$.
%Observe that if $n=k$, the tuple ${\bf t}[v]$ is simply a rearrangement of the elements of the tuple ${\bf t}$.
\item An \emph{$n$-ary selection condition} $\theta$ is an expression built by  using (some  or all of) the variables $x_1\ldots,x_n$, equalities $x_i=x_j$, negated equalities $x_i\neq x_j$, conjunctions, and disjunctions, where $1\leq i,j\leq n$.  Every $n$-ary selection condition $\theta$ gives rise to an  $n$-ary \emph{selection function}
$P_\theta: D^n\rightarrow \{0,1\}$ such that for every $n$-tuple ${\bf t}=(d_1,\ldots,d_n)\in D^n$, we have that $P_\theta({\bf t})=1$, if the assignment $x_i\mapsto d_i$, $1\leq i \leq n$, satisfies $\theta$, while $P_\theta({\bf t})=0$, otherwise. 

Note that every $n$-ary selection condition is also an $m$-ary selection condition, for every $m\geq n$. 
\end{itemize}

\begin{Def}[Operations on $\mathbb{K}$-relations]\label{rel-op-def} Let\/ $\mathbb{K}=(K,+,\cdot, - , s, 0,1)$ be a semiring with monus and support. We define the following operations on $\mathbb{K}$-relations.
\begin{itemize}
\item {\bf Union:}   If $R_1$ and $R_2$ are two $n$-ary $\mathbb{K}$-relations, then the \emph{union}  $R_1 \cup R_2$ is the $n$-ary $\mathbb{K}$-relation defined by   
$(R_1 \cup R_2)({\bf t}) = R_1({\bf t}) + R_2({\bf t})$, for all ${\bf t} \in D^n$.
\item {\bf Difference:}
If $R_1$ and $R_2$ are two $n$-ary $\mathbb{K}$-relations, then the \emph{difference}  $R_1 \setminus R_2$ is the $n$-ary $\mathbb{K}$-relation defined by   
$(R_1 \setminus R_2)({\bf t}) = R_1({\bf t}) - R_2({\bf t})$, for all ${\bf t} \in D^n$.
\item {\bf Cartesian Product:} 
If $R_1$ is an $m$-ary $\mathbb{K}$-relation and  $R_2$ is an $n$-ary $\mathbb{K}$-relation, then the \emph{Cartesian product}  $R_1 \times R_2$ is the $(m+n)$-ary $\mathbb{K}$-relation defined by   
$(R_1 \times R_2)({\bf t}) = R_1({\bf t}[v_1]) \cdot R_2({\bf t}[v_2])$, for all ${\bf t} \in D^{m+n}$,
where
$v_1=(1,\ldots,m)$ and $v_2=(m+1,\ldots,m+n)$.
\item {\bf Projection:}
\begin{itemize}
    \item If $R$ is an $n$-ary $\mathbb{K}$-relation with $n\geq 1$ and $v=(i_1,\ldots,i_k)$   is a non-empty sequence  of distinct indices from $\{1,\ldots,n\}$, then the
\emph{projection} 
$\pi_v(R)$ is the $k$-ary $\mathbb{K}$-relation defined by
$$\pi_v(R)({\bf t}) = \sum_{{\bf w}[v]={\bf t}~ \text{and}~R({\bf w}) \neq 0} R({\bf w}), $$
 for all ${\bf t}\in D^n$.
(We take the empty sum to be  $0$.)

\item If $R$ is a $\mathbb{K}$-relation, then  the \emph{projection of $R$ on the empty sequence}
is the $0$-ary relation 
$\pi_{[  ]>}(R)$ defined by $\pi_{[  ] >}(R)(< >) = \sum_{R({\bf t})\neq 0}R({\bf t})$.
\end{itemize}

\item {\bf Selection:}
If $R$ is an $n$-ary $\mathbb{K}$-relation and  $\theta$ is an $n$-ary selection condition, then the \emph{selection $\sigma_\theta(R)$} is the $n$-ary $\mathbb{K}$-relation defined by $\sigma_\theta(R) ({\bf t}) = R({\bf t}) \cdot  P_\theta ({\bf t})$, for all ${\bf t}\in D^n$, where $P_\theta$ is the $n$-ary selection function arising from $\theta$.

 Note that $\sigma_\theta(R)$ is indeed a $\mathbb K$-relation, since $\mathrm{supp}(\sigma_\theta(R))\subseteq \mathrm{supp}(R)$, and $\mathrm{supp}(R)$ is a finite set.

\item {\bf Support:}
If $R$ is an $n$-ary $\mathbb{K}$-relation, then the \emph{support $\mathrm{supp}(R)$ of $R$} is the $n$-ary $\mathbb{K}$-relation defined by 
$\mathrm{supp}(R) ({\bf t})= \mathsf{s} (R({\bf t}))$, for all ${\bf t}\in D^n$.

% $${\color{red}{???: R_1 \div R_2({\bf t}) =  \prod_{{\bf t}'= ({\bf t}, {\bf s}) \ \text{and}  \ R_2({\bf s}) \neq 0} (R_1({\bf t}') \cdot  R_2({\bf s}))}}$$
\end{itemize}
\end{Def}

%\guillermo{Observe that if $v = [ \ ]$, and we are considering the Boolean semiring $\mathbb{B}$, we have that $$\pi_v(R)([ \ ]) = \sum_{{\bf w}[v]=[ \ ]~ \text{and}~R({\bf w}) \neq 0} R({\bf w}), $$
%which gives $1$ if $\mathrm{supp}(R)\neq \emptyset$ and $0$ if
%$\mathrm{supp}(R)= \emptyset$. Thus, if $\mathrm{supp}(R)\neq \emptyset$, $\mathrm{supp}(\pi_v(R))= \{ [ \ ]\}$ and if $\mathrm{supp}(R)= \emptyset$ then $\mathrm{supp}(\pi_v(R))= \emptyset$, following the practice in~\cite{AbiVi}.
%}

In Definition~\ref{rel-defn}, we defined the support $\supp(R)$ of an $n$-ary $\mathbb{K}$-relation as the set of all $n$-tuples  in $D^n$ on which $R$ takes a non-zero value, while here we just defined $\supp(R)$ to be another $\mathbb{K}$-relation. There is no inconsistency between these two definitions, because $\supp(R)$ viewed as a $\mathbb{K}$-relation takes values $0$ and $1$ only, so it coincides with the ordinary relation consisting of all $n$-tuples on which $R$ takes a non-zero value. In what follows, we will use these two views of $\supp(R)$ interchangeably. 

We now introduce the syntax of basic relational algebra, which is a formalism for expressing combinations of the operations on $\mathbb{K}$-relations that were described in Definition~\ref{rel-op-def}.  As usual, a \emph{relational schema} or, simply, a \emph{schema} is a finite sequence $\tau= (\hat{R_1}, \ldots, \hat{R_p})$ of relation symbols, where each relation symbol $\hat{R_i}$ has a designated positive integer $r_i$ as its \emph{arity}.

\begin{Def}[Basic Relational Algebra] \label{bradef} Let
$\tau= (\hat{R_1},\ldots,\hat{R_p})$ be a schema. The \emph{basic relational algebra language $\mathcal{BRA}(\tau)$ over $\tau$}  is the set of all expressions defined recursively  as follows:
\begin{itemize}
\item   Each relation symbol $\hat{R_i}$ is a $\mathcal{BRA}(\tau)$-expression of arity $r_i$, 
 $1\leq i\leq p$.
\item If  $E_1, E_2$ are $\mathcal{BRA}(\tau)$-expressions of the  same arity $n$, then $E_1 \cup E_2 $ and $E_1\setminus E_2$ are $\mathcal{BRA}(\tau)$-expressions of arity $n$. 
\item If $E_1, E_2$ are $\mathcal{BRA}(\tau)$-expressions of  arities $n, m$, respectively,  then $E_1 \times E_2 $ is a $\mathcal{BRA}(\tau)$-expression of  arity $n+k$.

\item If $E$ is a $\mathcal{BRA}(\tau)$-expression of  arity $n\geq 1$  and $v=(i_1,\ldots,i_k)$ is a sequence of distinct indices from $\{1,\ldots,n\}$,
then $\pi_v(E) $ is a $\mathcal{BRA}(\tau)$-expression of arity $k$. Furthermore, if $E$ is a $\mathcal{BRA}(\tau)$-expression of arity $n\geq 0$, then $\pi_{[]}(E)$ is a $\mathcal{BRA}(\tau)$-expression of arity $0$.

\item if $E$ is a $\mathcal{BRA}(\tau)$-expression with  arity $n$ and $\theta$ is an $n$-ary selection condition,  then $\sigma_\theta(E) $ is a $\mathcal{BRA}(\tau)$-expression of  arity $n$,

\item if $E$ is a $\mathcal{BRA}(\tau)$-expression of  arity $n$, then $\supp(E) $ is a $\mathcal{BRA}(\tau)$-expression of arity $n$.

\end{itemize}
\end{Def}

\begin{Def} [Semantics of Basic Relational Algebra] \label{algebraop}
Let $\mathbb{K}=(K,+,\cdot, -, \mathsf{s}, 0,1)$ be a semiring with monus and support, and let  $\tau= (\hat{R_1}, \ldots,\hat{R_p})$ be a schema.
\begin{itemize}
\item A  \emph{$\mathbb{K}$-database over $\tau$} is a finite sequence $\mathbf{I}=(R_1,\ldots,R_p)$ of\/ $\mathbb{K}$-relations such that the arity of each $\mathbb{K}$-relation $R_i$ matches the arity of the  relation symbol $\hat{R_i}$, $1\leq i\leq p$.  

We say that $\mathbf I$ is \emph{non-trivial}  
if at least one of the relations $R_i$ of $\mathbf I$
has non-empty support.

In what follows, we make the blanket assumption that all 
$\mathbb K$-databases considered are non-trivial.

\item If $E$ is a $\mathcal{BRA}(\tau)$-expression of arity $n$ and\/ $\mathbf{I}= (R_1,\ldots,R_p)$ is a $\mathbb{K}$-database, then $E^\mathbf{I}$ is the $n$-ary $\mathbb{K}$-relation obtained by evaluating $E$ on $\mathbf{I}$ according to the definitions of the operations on $\mathbb{K}$-relations. For the sake of completeness, we spell out the precise definition of  the semantics of $\mathcal{BRA}$-expressions:
\begin{itemize}
    \item $(\hat{R_i})^\mathbf{I}= R_i$,
    where $1\leq i\leq p$. 
    \item $(E_1\cup  E_2)^\mathbf{I}= E_1^\mathbf{I}\cup E_2^\mathbf{I}$
    and $(E_1\setminus   E_2)^\mathbf{I}= E_1^\mathbf{I}\setminus E_2^\mathbf{I}$,  where $E_1$ and $E_2$ are $\mathcal{BRA}$-expressions of the same arity.
    \item $(E_1\times  E_2)^\mathbf{I}= E_1^\mathbf{I}\times E_2^\mathbf{I}$.
    \item $(\pi_v(E))^\mathbf{I}= \pi_v(E^\mathbf{I}) $ and $(\pi_{[]}(E))^\mathbf{I}= \pi_{[]}(E^\mathbf{I})$.
    \item $(\sigma_\theta(E))^\mathbf{I}=
    \sigma_\theta(E^\mathbf{I})$.
    \item $(\supp(E))^\mathbf{I}= \supp(E^\mathbf{I})$.
    
\end{itemize}
\end{itemize}
\end{Def}

To simplify the notation in relational algebra expressions, we will often use $R_i$ to denote a relation symbol, instead of $\hat{R_i}$. Also, when the schema $\tau$ is understood from the context, we will write
$\mathcal{BRA}$, instead of
$\mathcal{BRA}(\tau)$, and will talk about $\mathcal{BRA}$-expressions, instead of $\mathcal{BRA}(\tau)$-expressions.

\begin{Rmk} [Necessity of Empty Projections]
Codd~\cite{Codd} did not allow for empty projections $\pi_{[]}(R)$ in the definition of relational algebra. If empty projections are disallowed, then all relational algebra expressions have positive arity, thus they cannot express $0$-ary queries; in particular, they cannot express queries expressible by first-order sentences. In subsequent expositions of relational algebra (e.g., in~\cite{AbiVi}), empty projections are allowed, which is the choice we made here.

\end{Rmk}

\begin{Rmk}[Necessity of Support]\label{rem:support}
In Codd's original paper on the relational data model~\cite{DBLP:journals/cacm/Codd70}, the support operation was not among the operations on relations considered. There is a good reason for this: on the Boolean semiring, we have that $\supp(E)=E$ holds for every relational algebra expression $E$.
On an arbitrary semiring $\mathbb{K}$, though,  the support is needed as a separate operation, because, in general, the support is not expressible using the other relational algebra operations.
For example, 
it is easy to show that, for the fuzzy semiring
$\mathbb{F}=([0, 1], \max, \min,0,1)$, the support operation cannot be expressed in terms of the other five relational algebra operations. Indeed, assume that $\tau$ consists of a single relation symbol $R$ of arity $1$. Let $\mathbf{I}$ be the $\mathbb{F}$-database in which $R$ is the  $\mathbb{F}$-relation with  $R(a)=0.5$ and $R(b)=0$, for all $b\neq a$.  A straightforward induction shows that if $E$ a relational algebra expression (not involving  support) over $\tau$, then  $E(\mathbf{I})$ is an $\mathbb{F}$-relation such that $E^\mathbf{I}({\bf t}) \in [0,0.5]$ if ${\bf t} = (a,\ldots,a)$, and $E^\mathbf{I}({\bf t})=0$ if ${\bf t}\neq (a,\ldots,a)$. Thus, $E^\mathbf{I}\neq \supp(R)$.  This shows that $\mathsf{s}$ is not definable from the other operations in the fuzzy semiring.
% Now let us show that it remains undefinable in the presence of monus with an argument based on the congruences of the algebra (equivalence relations compatible with the operations). On the one hand, the algebra without $\mathsf{s}$ is not simple, i.e.\ it has non-trivial congruences such as $\{(a,b) \mid a - b, b - a \leq \varphirac{1}{2}\}$. On the other hand, the algebra expanded with $\mathsf{s}$ is simple (its only congruences are the identity and the total) and hence $\mathsf{s}$ cannot be definable. Indeed, suppose that $\theta$ is a congruence and $(a,b) \in \theta$ with $a \neq b$. Let us assume w.l.o.g.\ that $a>b$. Then, $(a - a, a - b) = (0, a) \in \theta$, and so $(\mathsf{s}(0),\mathsf{s}(a)) = (0,1) \in \theta$. Now, for each $x \in [0,1]$, we have $(0 \lor x, 1 \lor x) = (x,1) \in \theta$. Therefore, $\theta$ is the total congruence.
\end{Rmk}

\subsection{Basic Relational Calculus}

In this section, we introduce a first-order language called \emph{basic relational calculus}, which has existential quantification and  
a limited form of negation, but not universal quantification. 
Let $\tau=(\hat{R_1},\ldots \hat{R_p})$ be a schema. The building blocks of the \emph{basic relational calculus over $\tau$}, denoted by $\mathcal{BRC}(\tau)$, consist of a countable set $V=\{x_1,\ldots, x_n, \ldots \}$ of variables, 
a distinguished binary relation $=$ (equality), the relation symbols of $\tau$,
three binary connectives $\wedge$ (`and'), $\vee$ (`or') and $\butnot$ (`but not'), 
a unary connective $\nabla$ (`it is not false that'), and the existential  quantifier $\exists$ (`there exists'). 
In what follows, we will  write
$y_1,\ldots, y_i,\ldots$ to denote \emph{meta-variables}, i.e., each $y_i$ is one of the variables in the set $V$. 

We begin by rigorously defining the  syntax of 
$\mathcal{BRC}(\tau)$-formulas and also the notion of  the free variables of a $\mathcal{BRC}(\tau)$-formula. 
\begin{Def}  \label{brc-synt-def}
Let $\tau =(\hat{R_1},\ldots \hat{R_p})$ be a schema. The \emph{basic relational calculus language $\mathcal{BRC}(\tau)$ over $\tau$}  is the set of all expressions defined recursively  as follows:
\begin{itemize}
\item If $y_i$ and $y_j$ are variables, then
the expression $y_i=y_j$ is a $\mathcal{BRC}(\tau)$-formula with $y_i$ and $y_j$ as its free variables.
\item If $\hat{R_i}$ is a relation symbol  of arity $r_i$ and if $y_1,\ldots,y_{r_i}$ are variables, then
the expression $\hat{R_i}(y_1,\ldots,y_{r_i})$ is a $\mathcal{BRC}(\tau)$-formula with $y_1,\ldots,y_{r_i}$ as its free variables.
\item If $\varphi$ is a $\mathcal{BRC}(\tau)$-formula with $y_1,\ldots,y_n$ as its free variables and if $\psi$ is a $\mathcal{BRC}(\tau)$-formula with $y_{n+1},\ldots,y_{n+k}$ as its free variables, then the expressions
$\varphi \vee \psi$, $\varphi \wedge \psi$, and $\varphi \butnot \psi$ are $\mathcal{BRC}(\tau)$-formulas, and each of them has $y_1,\ldots,y_n,y_{n+1},\ldots,y_{n+k}$ as its free variables. Note that   the sets $\{y_1,\ldots,y_n\}$ and $\{y_{n+1},\ldots, y_{n+k}\}$ need not be disjoint, so when we write $y_1,\ldots,y_n,y_{n+1},\ldots,y_{n+k}$ we assume that repeated occurrences have been eliminated from the list.
\item  If $\varphi$ is a $\mathcal{BRC}(\tau)$-formula with $y_1,\dots,y_n$ as its free variables,
then the expression $\nabla \varphi$ is
 a $\mathcal{BRC}(\tau)$-formula with $y_1,\dots,y_n$ as its free variables.
\item If $\varphi$ is a $\mathcal{BRC}(\tau)$-formula with $y_1,\dots,y_n$ as its free variables, then for every $i$ with $1\leq i\leq n$, then the expression
$\exists y_i\varphi$ is a  $\mathcal{BRC}(\tau)$-formula with the variables $y_1,\ldots,y_n$ other than $y_i$ as its free variables.
\end{itemize}
From now on,  the notation $\varphi(y_1,\ldots,y_n)$  denotes a $\mathcal{BRC}(\tau)$-formula $\varphi$ whose free variables are $y_1,\ldots,y_n$. 
\end{Def}

\begin{Def}[$\mathbb{K}$-structures and assignments] Let\/ $\mathbb{K}=(K,+,\cdot,-,\mathsf{s},0,1)$ be a semiring with monus and support, and 
let  $\tau= (\hat{R_1}, \ldots, \hat{R_p})$ be a schema.
\begin{itemize}
\item A \emph{$\mathbb{K}$-structure over $\tau$} is a tuple $\mathbf{A}= (A,R_1, \dots, R_p)$, 
where $A$ is a non-empty subset of $D$, called the \emph{universe} of $\mathbf{A}$,
and each $R_i$ is a $\mathbb{K}$-relation on $A$ 
whose arity matches the arity  of the  relation symbol $\hat{R_i}$,
$1\leq i\leq p$.

We say that $\mathbf{A}$ is a  \emph{finite $\mathbb{K}$-structure} if its  universe $A$  is a finite set.

\item An \emph{assignment} on $\mathbf{A}$ is a function $\alpha\colon V\longrightarrow A$, assigning to every variable  an element in $A$.
\item If $\alpha$ is an assignment on $\mathbf{A}$ and $b$ is an element of $A$, then we write $\alpha[x_i/b]$ for the assignment $\beta$ on $\mathbf{A}$ such that
$\beta(x_i)=b$ and $\beta(x_j) = \alpha(x_j)$, for every variable $x_j \neq  x_i$. 

\end{itemize}
\end{Def}

%Let $\mathbb{K}= (K, +, \cdot, -, \mathsf{s}, 0,1)$ be a semiring with monus and support.

In what follows, we make the blanket assumption that all $\mathbb{K}$-structures considered are finite, unless explicitly stated otherwise.
We are now ready to give semiring semantics to $\mathcal{BRC}(\tau)$-formulas.

\begin{Def}\label{FOL-sem-def}
Let\/ $\mathbb{K}=(K,+, \cdot, -, \mathsf{s}, 0,1)$ be a semiring with monus and support,  
let $\tau=(\hat{R_1},\ldots, \hat{R_p})$ be a schema, and let $\varphi$ be a $\mathcal{BRC}(\tau)$-formula.
The \emph{semantics of $\varphi$ on $\mathbb{K}$} is a function $\|\varphi\|$ that takes as input a $\mathbb{K}$-structure  $\mathbf{A}=(A,R_1,\ldots,R_p)$ and an assignment  $\alpha$ on $\mathbf{A}$, and returns   a  value $\|\varphi\|(\mathbf{A}, \alpha)$ in $K$ that is defined recursively and simultaneously for all assignments as follows: 
\begin{itemize}

%\item If $\varphi$ is a $\mathcal{FO}(\tau)$-formula of the form $x_i= x_j$, then for every assignment $\alpha$ on $\mathbf{A}$, we put $\varphi(\mathbf{A},\alpha)=
\item $\|x_i=x_j\|(\mathbf{A},\alpha) = 
\begin{cases} 1 & \text{if } \alpha(x_i) = \alpha(x_j) ; \\ 0 & \mbox{otherwise.} \end{cases}$

\item $ \|\hat{R_i} (x_{i_1},\ldots,x_{i_n})\|(\mathbf{A}, \alpha) = R_i(\alpha(x_{i_1}), \ldots \alpha(x_{i_n}))$, for $1\leq i\leq p$.

\item $\|\varphi_1 \wedge \varphi_2\|(\mathbf{A},\alpha) = \|\varphi_1\|(\mathbf{A}, \alpha) \cdot \| \varphi_2\|(\mathbf{A}, \alpha)$.

\item $\|\varphi_1 \vee \varphi_2\|(\mathbf{A},\alpha) = \|\varphi_1\|(\mathbf{A}, \alpha)+ \|\varphi_2\|(\mathbf{A}, \alpha)$.

\item $|\varphi_1 \butnot  \varphi_2\|(\mathbf{A},\alpha) = \|\varphi_1\|(\mathbf{A}, \alpha) - \| \varphi_2\|(\mathbf{A}, \alpha)$.

\item  $\|\nabla \varphi \|(\mathbf{A},\alpha) = \mathsf{s} (\|\varphi\|(\mathbf{A}, \alpha))  $.

\item $\|\exists x_i\varphi\| (\mathbf{A}, \alpha) = \sum_{b\in A}\| \varphi\|(\mathbf{A}, \alpha[x_i/b])$.
\end{itemize}

\end{Def}

Note that the semantics of the existential quantifier is well defined because of our blanket assumption that $\mathbf{A}$ has a finite universe. The semantics is still meaningful for $\mathbb{ K}$-structures with infinite universes, provided the semiring $\mathbb{K}$ considered has additional closure properties. In particular, this holds true when $\mathbb{K}$ is a complete bounded distributive lattice. 

Using a straightforward induction, one can show that the value
$\|\varphi\|(\mathbf{A},\alpha)$ depends only on the values of the assignment $\alpha$ on the free variables of $\varphi$. More precisely, we have the following lemma.

\begin{Lem}[Relevance Lemma] \label{rel-lem} 
Let $\tau$ be a schema and $\varphi(y_1,\ldots,y_n)$ be a $\mathcal{BRC}(\tau)$-formula whose free variables are $y_1,\ldots,y_n$, and let $\mathbf{A}$ be a $(\tau)$-structure. If $\alpha$ and $\beta$ are two assignments on $\mathbf{A}$ such that
$\alpha(y_i)= \beta(y_i)$ for every free variable $y_i$ of $\varphi$, then
$\|\varphi\|(\mathbf{A},\alpha)= 
\|\varphi\|(\mathbf{A},\beta)$.
\end{Lem}

  If $\varphi(y_1,\ldots,y_n)$ is a formula and $(a_1,\ldots,a_n)$ is a sequence of elements from the universe $A$ of a $\mathbb{K}$-structure $\mathbf{A}$, then we will write
 $\|\varphi\|(\mathbf{A},(a_1,\ldots,a_n))$ to denote the value $\|\varphi\|(\mathbf{A}, \alpha)$, where $\alpha$ is any assignment on $\mathbf{A}$ such that
$\alpha(y_1)=a_1,\ldots,\alpha(y_n)=a_n$. By the preceding Relevance Lemma, this value is well defined. Furthermore, if $\varphi$ is a \emph{sentence} (i.e., $\varphi$ has no free variables), we will write $\|\varphi\|(\mathbf{A},< >)$ to denote $\|\varphi\|(\mathbf{A}, \alpha)$, where $\alpha$ is an arbitrary assignment; note that in this case, $\|\varphi\|(\mathbf{A},\alpha)$ is the same value, independently of the assignment $\alpha$.

\begin{Def} \label{brc-eval-def}
Let\/ $\mathbb{K}$ be a semiring with monus and support, and let $\tau$ be a schema.
\begin{itemize}
    \item
If $\varphi(y_1,\ldots,y_n)$ is a $\mathcal{BRC}(\tau)$-formula whose free variables are $y_1,\ldots,y_n$  and if $\mathbf{A}$ is a $\mathbb{K}$-structure, then $\varphi^\mathbf{A}\colon D^n\longrightarrow K$ is the $n$-ary $\mathbb{K}$-relation defined as follows:

$\varphi^\mathbf{A}(a_1,\ldots,a_n) =
\begin{cases} \|\varphi\|(\mathbf{A},(a_1,\ldots,a_n))& \mbox{if $(a_1,\ldots,a_n)\in A^n$} \\
0 & \mbox{if $(a_1,\ldots,a_n) \notin A^n$} \end{cases}$

Note that $\varphi^\mathbf{A}$ is actually a $\mathbb{K}$-relation on $A$ since, obviously, $\supp(\varphi^\mathbf{A})\subseteq A^n$. 

\item If $\varphi$ is a $\mathcal{BRC}(\tau)$-sentence  and if $\mathbf{A}$ is a $\mathbb{K}$-structure, then $\varphi^\mathbf{A}\colon D^0\longrightarrow K$ is the $0$-ary $\mathbb{K}$-relation such that  $\varphi^{\mathbf A}(<>) =   \|\varphi\|(\mathbf{A},<>)$.
\end{itemize}
\end{Def}

The next proposition describes some useful properties of the semantics of $\mathcal{BRC}$-formulas that will be used in the proofs of the main results.

\begin{Pro}\label{pro:basic}
Let $\tau=(\hat{R_1},\ldots \hat{R_p})$ be a schema. The following statements are true.
\begin{enumerate}
\item Consider the atomic formula
     $\hat{R_i}(y_1,\ldots,y_n)$, where $\hat{R_i}$ is one of the relation symbols of the schema $\tau$,  $n$ is the arity of $\hat{R_i}$, and the $y_i$'s are distinct variables.
    If $\mathbf{A}=(A,R_1\ldots,R_p)$ is a $\mathbb{K}$-structure, then   $(\hat{R_i}(y_1,\ldots,y_n))^\mathbf{A}= R_i$.
    \item Let $\varphi_1(y_1,\ldots,y_n)$ and $\varphi_2(y_1,\ldots,y_n)$ be two $\mathcal{BRC}(\tau)$-formulas with the same free variables. If  $\mathbf{A}=(A,R_1\ldots,R_p)$ is a $\mathbb{K}$-structure, then
    \begin{enumerate}
\item $(\varphi_1 \vee \varphi_2)^\mathbf{A} = \varphi_1^\mathbf{A} \cup \varphi_2^\mathbf{A}$.
\item  $(\varphi_1 \butnot \varphi_2)^\mathbf{A} = \varphi_1^\mathbf{A} \setminus \varphi_2^\mathbf{A}$.
\end{enumerate}
 
\item Let $\varphi_1(y_1,\ldots,y_n)$ and $\varphi_2(y_{n+1},\ldots,y_{n+m})$ be two $\mathcal{BRC}(\tau)$-formulas, where
$y_1,\ldots,y_{n+m}$ are distinct variables. If $\mathbf{A}$ is a $\mathbb{K}$-structure, then
$(\varphi_1\wedge \varphi_2)^\mathbf{A}=
\varphi_1^\mathbf{A}\times \varphi_2^\mathbf{A}$.
\item Let $\varphi_1(y_1,\ldots,y_n)$ be a $\mathcal{BRC}(\tau)$-formula. If $\mathbf{A}=(A,R_1\ldots,R_p)$ is a $\mathbb{K}$-structure, then $(\nabla \varphi)^\mathbf{A} = \supp (\varphi^\mathbf{A}) $.
\item Let $\varphi(y_1,\ldots,y_n)$ be a $\mathcal{BRC}(\tau)$-formula, let $i$ be an index with $1\leq i\leq n$, and let $v=(i_1,\ldots,i_{n-1})$ be the sequence of distinct indices from $\{1,\ldots,n\}\setminus \{i\}$ ordered by the standard ordering of the natural numbers.
If $\mathbf{A}=(A,R_1\ldots,R_p)$ is a $\mathbb{K}$-structure, then $(\exists y_i \varphi)^\mathbf{A} = \pi_v(\varphi^\mathbf{A})$. 
     \end{enumerate}
\end{Pro}
\begin{proof}
The first part about the atomic formula  $\hat{R_i}(y_1,\ldots,y_n)$ follows easily from Definitions~\ref{FOL-sem-def} and~\ref{brc-eval-def}.

For the second part, let
$\varphi_1$ and $\varphi_2$ be two $\mathcal{BRC}(\tau)$-formulas with $y_1,\ldots,y_n$ as their free variables, and let $\mathbf{A}$ be a $\mathbb{K}$-structure. 
To show that
 $(\varphi_1 \vee \varphi_2)^\mathbf{A} = \varphi_1^\mathbf{A} \cup \varphi_2^\mathbf{A}$, we have to show that for every tuple $(a_1,\ldots,a_n)\in D^n$, we have 
 $(\varphi_1 \vee \varphi_2)^\mathbf{A}(a_1,\ldots,a_n) = (\varphi_1^\mathbf{A} \cup \varphi_2^\mathbf{A})(a_1,\ldots,a_n) $. We distinguish two cases for the tuple $(a_1,\ldots,a_n) \in D^n$, namely, whether or not $(a_1,\ldots,a_n) \in A^n$, where $A$ is the universe of $\mathbf{A}$.

 If $(a_1,\ldots,a_n)\in A^n$, then we have that
$$ (\varphi_1 \vee \varphi_2)^\mathbf{A}(a_1,\ldots,a_n)=
\|\varphi_1 \vee \varphi_2\|(\mathbf{A},(a_1,\ldots,a_n)) = 
\|\varphi_1\|
  (\mathbf{A}, (a_1,\ldots,a_n))+
  \| \varphi_2\| (\mathbf{A}, (a_1,\ldots,a_n)) =$$
  $$\varphi_1^\mathbf{A}(a_1,\ldots,a_n) + \varphi_2^\mathbf{A}(a_1,\ldots,a_n)= (\varphi_1^\mathbf{A} \cup \varphi_2^\mathbf{A})(a_1,\ldots,a_n),$$
  where the first equation follows from Definition~\ref{brc-eval-def}, the second  follows from the semantics of $\mathcal{BRC}(\tau)$-formulas in Definition~\ref{FOL-sem-def}, the third follows from Definition~\ref{brc-eval-def}, and the last  follows from the definition of operations on $\mathbb{K}$-relations in Definition~\ref{rel-op-def}.
  
   If $(a_1,\ldots,a_n) \not \in A^n$, then $ (\varphi_1 \vee \varphi_2)^\mathbf{A}(a_1,\ldots,a_n)=0$ and also
   $ \varphi_1^\mathbf{A}(a_1,\ldots,a_n)= 0 =  \varphi_1^\mathbf{A}(a_1,\ldots,a_n)$ by Definition~\ref{brc-eval-def}.  Therefore,
  $$(\varphi_1 \vee \varphi_2)^\mathbf{A}(a_1,\ldots,a_n) = 0 = \varphi_1^\mathbf{A}(a_1,\ldots,a_n)  + \varphi_2^\mathbf{A}(a_1,\ldots,a_n) 
  = (\varphi_1^\mathbf{A} \cup \varphi_2^\mathbf{A})(a_1,\ldots,a_n) .$$ 
   The proof that
 $(\varphi_1 \butnot \varphi_2)^\mathbf{A} = \varphi_1^\mathbf{A} \setminus \varphi_2^\mathbf{A}$ is entirely analogous to the preceding one, but with the connective $\butnot$ in the role of the connective $\vee$, and with the connective $\setminus$
 in  the role of the connective $\cup$. This completes the proof of the second part.

 For the third part, assume that $\varphi_1(y_1,\ldots,y_n)$ and $\varphi_2(y_{n+1},\ldots,y_{n+m})$ are two $\mathcal{BRC}(\tau)$-formulas, where
$y_1,\ldots,y_{n+m}$ are distinct variables %(i.e., the two formulas have no free variables in common), 
and let $\mathbf{A}$ be a $\mathbb{K}$-structure.  To show that
$(\varphi_1\wedge  \varphi_2)^\mathbf{A}=
\varphi_1^\mathbf{A}\times \varphi_2^\mathbf{A}$, we have to show that for every tuple $(a_1,\ldots,a_{n+m})\in D^n$, we have that $$(\varphi_1\wedge \varphi_2)^\mathbf{A}(a_1,\ldots,a_{n+m})=
(\varphi_1^\mathbf{A}\times \varphi_2^\mathbf{A})(a_1,\ldots,a_{n+m}).$$ We  distinguish two cases for the tuple $(a_1,\ldots,a_n) \in D^n$, namely, whether or not $(a_1,\ldots,a_{n+m}) \in A^{n+m}$, where $A$ is the universe of $\mathbf{A}$. 
 
 If $(a_1,\ldots,a_{n+m})\in A^{n+m}$, then 
 $$(\varphi_1\wedge \varphi_2)^\mathbf{A}(a_1,\ldots,a_{n+m})=
 \|\varphi_1\wedge \varphi_2\|(\mathbf{A}, (a_1,\ldots,a_{n+m}))= $$
$$\|\varphi_1\|(\mathbf{A}, (a_1,\ldots,a_{n})) \cdot 
\|\varphi_2\|(\mathbf{A}, (a_{n+1},\ldots,a_{n+m}))=$$
$$ 
\varphi^\mathbf{A}(a_1,\ldots,a_n) \cdot
\varphi^\mathbf{A}(a_{n+1},\ldots,a_{n+m})= (\varphi_1^\mathbf{A}\times \varphi_2^\mathbf{A})(a_1,\ldots,a_{n+m})
,$$
where the first equation follows from Definition~\ref{brc-eval-def},
the second  follows 
the semantics of $\mathcal{BRC}(\tau)$-formulas in Definition~\ref{FOL-sem-def} and  from the Relevance Lemma~\ref{rel-lem}, the third  follows from Definition~\ref{brc-eval-def}, and the last  follows from the definition of the operations on $\mathbb{K}$-relations in Definition~\ref{rel-op-def}.

If $(a_1,\ldots,a_{n+m})\not \in A^{n+m}$, then $(\varphi_1\wedge \varphi_2)^\mathbf{A}(a_1,\ldots,a_{n+m})=0$ 
 by Definition~\ref{brc-eval-def}. 
Furthermore, we  have that
$(a_1,\dots,a_n)\not \in A^n$ or
$(a_{n+1},\ldots,a_{n+m}) \not \in A^m$. By Definition~\ref{brc-eval-def}, we have that
$\varphi_1^\mathbf{A}(a_1,\ldots,a_n)=0$
or $\varphi_2^\mathbf{A}(a_{n+1},\ldots,a_{n+m})=0$. In either case, we have that 
$$0= \varphi^\mathbf{A}(a_1,\ldots,a_n) \cdot
\varphi^\mathbf{A}(a_{n+1},\ldots,a_{n+m})= (\varphi_1^\mathbf{A}\times \varphi_2^\mathbf{A})(a_1,\ldots,a_{n+m}),$$  hence
$(\varphi_1\wedge \varphi_2)^\mathbf{A}(a_1,\ldots,a_{n+m})=
(\varphi_1^\mathbf{A}\times \varphi_2^\mathbf{A})(a_1,\ldots,a_{n+m})$. 
This completes the proof of the third part.

For the fourth part, assume that $\varphi(y_1,\ldots,y_n)$ is a $\mathcal{BRC}(\tau)$-formula and   $\mathbf{A}=(A,R_1\ldots,R_p)$ is a $\mathbb{K}$-structure. To show that  $(\nabla \varphi)^\mathbf{A} = \supp (\varphi^\mathbf{A}) $ we have   to show that for every tuple $(a_1,\ldots,a_{n})\in D^n$, we have that
$(\nabla \varphi)^\mathbf{A}(a_1,\ldots,a_{n}) = \supp (\varphi^\mathbf{A})(a_1,\ldots,a_{n}) $. We distinguish two cases for the tuple $(a_1,\ldots,a_n) \in D^n$, namely, whether or not $(a_1,\ldots,a_n) \in A^n$, where $A$ is the universe of $\mathbf{A}$.  If $(a_1,\ldots,a_n) \in D^n$, then, by Definition~\ref{brc-eval-def}, 
 $$(\nabla \varphi)^\mathbf{A} (a_1,\ldots,a_n)=\begin{cases} 1 & \text{if }  \varphi^\mathbf{A}(a_1,\ldots,a_n) \neq 0 \\ 0 & \text{otherwise}. \end{cases} $$
 Thus, $(\nabla \varphi)^\mathbf{A}  (a_1,\ldots,a_n)= \supp(\varphi^\mathbf{A})(a_1,\ldots,a_n)$. If $(a_1,\ldots,a_n) \notin D^n$, then, by  Definition~\ref{brc-eval-def}, both $(\nabla \varphi)^\mathbf{A}  (a_1,\ldots,a_n)=0$ and $\varphi^\mathbf{A} (a_1,\ldots,a_n)  =0$, so  $\supp(\varphi^\mathbf{A})(a_1,\ldots,a_n) = 0 =(\nabla \varphi)^\mathbf{A}  (a_1,\ldots,a_n)$ as desired.

 For the fifth part, assume that 
   $\varphi(y_1,\ldots,y_n)$ is a $\mathcal{BRC}(\tau)$-formula,  $i$ is an index with $1\leq i\leq n$, and $v=(i_1,\ldots,i_{n-1})$ is the sequence of distinct indices from $\{1,\ldots,n\}\setminus \{i\}$ ordered by the standard  ordering of the natural numbers.
 Let 
 $\mathbf{A}=(A,R_1\ldots,R_p)$ is a $\mathbb{K}$-structure. To show that $(\exists y_i \varphi)^\mathbf{A} = \pi_v(\varphi^\mathbf{A})$,  we have to show that for 
every tuple  $(a_1,\ldots,a_{n-1})\in D^{n-1}$, we have that $(\exists y_i \varphi)^{\mathbf{A}}(a_1,\ldots,a_{n-1}) = \pi_v(\varphi)^\mathbf{A}(a_1,\ldots,a_{n-1})$.
 We distinguish two cases for the tuple $(a_1,\ldots,a_{n-1}) \in D^{n-1}$, namely, whether or not $(a_1,\ldots,a_{n-1}) \in A^{n-1}$. Suppose that $(a_1,\ldots,a_{n-1}) \in A^{n-1}$.   From Definition \ref{rel-op-def} of the operations on $\mathbb K$-relations, we have that 
 $$\pi_v(\varphi^\mathbf{A})(a_1,\ldots,a_{n-1}) =
 \sum_{{\bf b}[v]= (a_1,\ldots,a_{n-1})\ \ \text{and} \ {\varphi^\mathbf{A}({\bf b}) \neq 0}} \varphi^\mathbf{A}({\bf b}). $$
 Furthermore, we have that $$ (\exists y_i \varphi)^\mathbf{A}(a_1, \dots, a_{n-1})=
  \| \exists y_i \varphi \|({\bf A}, (a_1,\ldots a_{n-1})) = 
 \sum_{\substack{{\bf b} \ \in \ A^n \\ {\bf b}[v] = (a_1,\ldots,a_{n-1})}} \| \varphi \| (\mathbf{A},{\bf b})=  
 \sum_{\substack{{\bf b} \ \in \ A^n \\ {\bf b}[v] = (a_1,\ldots,a_{n-1})}} \varphi^\mathbf{A}({\bf b}),$$
 where the first and the third equations follow from Definition \ref{brc-eval-def} and the second equation follows 
 from the semantics of $\mathcal{BRC}(\tau)$-formulas  in Definition~\ref{FOL-sem-def}.  Observe  that $\varphi_{1}^\mathbf{A}({\bf b}) \neq 0$  only if ${\bf b}$ is a tuple in $A^n$ (since,  if ${\bf b}$ is not a tuple in $A^n$, then $\varphi_{1}^\mathbf{A}({\bf b}) = 0$  by Definition~\ref{brc-eval-def}). Thus, 
$$(\exists y_i \varphi)^\mathbf{A}(a_1, \dots, a_{n-1}) = \sum_{{\bf b}[v]= (a_1,\ldots,a_{n-1})\ \ \text{and} \ {\varphi}^\mathbf{A}({\bf b}) \neq 0} \varphi^\mathbf{A}({\bf b}) = \pi_v(\varphi^{\bf A})(a_1,\ldots,a_{n-1}). $$

Suppose now that $(a_1,\ldots,a_{n-1}) \notin A^{n-1}$. By Definition~\ref{FOL-sem-def}, we have  $ (\exists y_i \varphi)^\mathbf{A}(a_1, \dots, a_{n-1})= 0$. Also, if ${\bf b} \in D^n$ is a tuple such that ${\bf b}[v]= (a_1,\ldots,a_{n-1})$, then ${\bf b}\not \in A^n$, hence, by Definition ~\ref{FOL-sem-def} again,
$\varphi^{\bf A}({\bf b}) = 0$. It follows that  $\pi_v(\varphi_{1}^\mathbf{A})(a_1,\ldots,a_{n-1}) =0$ and so 
$(\exists y_i \varphi)^\mathbf{A}(a_1, \dots, a_{n-1}) =  \pi_v(\varphi^{\bf A})(a_1,\ldots,a_{n-1})$. This completes the proof of the fifth part.
\end{proof}

\begin{Rmk} \label{rmk:neg-equal}
Let $ x_i\neq x_j$ be the formula
$(x_i= x_i) \butnot (x_i= x_j)$. It is easy to verify that for every $\mathbb{K}$-structure $\mathbf{A}$ and every two elements $a$ and $b$ in the universe of $\mathbf{A}$, we have that\\
$\|x_i\neq x_j\|(\mathbf{A}, a, b) = \begin{cases}
1 & \mbox{if $a\neq b$ }\\
0 & \mbox{if $a=b$.}
\end{cases}$

\noindent Therefore, $\mathcal{BRC}$ can express negated equalities, a fact that we will use in the sequel.
\end{Rmk}

% We now discuss the role of the $\butnot$ connective in other logical formalisms.

\begin{Rmk}
In the case of the Boolean semiring $\mathbb{B}$, the connectives $\wedge, \vee, \neg$ are inter-definable with the connectives $\wedge, \vee, \butnot$. In particular, $\varphi \butnot \psi$ has the same semantics on the Boolean semiring as $\varphi \wedge \neg \psi$, while $\neg \varphi$ has the same semantics as $(x_1=x_1) \butnot \varphi$. 
%Thus, choosing one or the other makes no difference in Boolean logic. %Furthermore $\nabla$ has many possible definitions in Boolean logic, for example $\nabla \varphi$ has the same semantics as $ \varphi$.
This, however, does not hold for arbitrary semirings. Concretely, it is not true that on  every  complete bounded distributive lattice (see Example~\ref{latt}), 
the connectives $\wedge, \vee, \neg$  are inter-definable with the connectives  $\wedge, \vee, \butnot, \top$.  Examples where this does not hold can be extracted from~\cite{Mcc, Smiley1975-SMITIO-46} by considering the algebraic \emph{duals} (i.e.,  the algebras where one reverses the order and swaps meets with joins) of the  complete bounded distributive lattices presented there.
\end{Rmk}

% {\color{red} What is an  example of a particular distribute lattice where this happens?} {\color{blue} GB answer: I don't have an easy answer for this. The difficult answer is that one would have to take the "dual" (so we swap the direction of the lattice order and take the operations (corresponding to connectives) to have their dual meaning) of the algebra defined by Matrix III in the paper https://doi.org/10.2307/2268715
%  by McKinsey. }

\subsection{Basic Relational Algebra vs.\ Basic Relational Calculus}

In this section,  $\mathbb{K}=(K, +, \cdot,-, \mathsf{s}, 0,1)$ is a fixed  zero-sum-free semiring with monus and support, and  $\tau = (\hat{R_1}, \ldots, \hat{R_p})$ is  a fixed schema. All $\mathbb{K}$-databases and $\mathbb{K}$-structures considered are over $\tau$.

\begin{Def}[Active domain]
If\/ $\mathbf{I}=(R_1,\ldots,R_p)$ is a $\mathbb{K}$-database, then the \emph{active domain}  of\/ $\mathbf{I}$, denoted by $\mathrm{adom}(\mathbf{I})$, is the set of all elements of $D$ that occur in the support of at least one relation $R_i$ of\/ $\mathbf{I}$.
\end{Def}

Since we have made the blanket assumption that we only consider non-trivial $\mathbb K$-databases (see Definition \ref{algebraop}), 
it follows that the active domain $\mathrm{adom}({\mathbf I})$ is a non-empty set. Furthermore,
since every relation $R_i$, $1\leq i\leq p$, of $\mathbf{I}$ has finite support, it follows that $\mathrm{adom}(\mathbf{I})$ is a finite  subset of the domain $D$. Thus,  $\mathrm{adom}(\mathbf{I})$ can be identified with a unary $\mathbb{K}$-relation that takes values $0$ and $1$ only. As it turns out,  there is a $\mathcal{BRA}$-relational algebra expression that uniformly defines the active domain on every $\mathbb{K}$-database, as spelled out in the next proposition. This is a basic fact that will be repeatedly used in the sequel.

\begin{Pro} \label{pro:active} Let $\mathbb{K}$ be a zero-sum-free semiring with monus and support, and let $\tau$ be a schema. Then there is a $\mathcal{BRA}(\tau)$-expression $E_{\mathrm{adom}}$ such that 
for every  $\mathbb{K}$-database $\mathbf{I}$, we have that
$E_{\mathrm{adom}}^\mathbf{I}= \mathrm{adom}(\mathbf{I})$.
\end{Pro}
\begin{proof}
For concreteness, let us assume that $\tau$ consists of a single binary relation symbol $\hat{R}$.  The argument 
presented below extends easily to arbitrary schemas with only slight modifications.

Let $E_{\mathrm {adom}}$  be the $\mathcal{BRA}(\tau)$-relational algebra expression $\supp(\pi_1(\hat{R}) \cup \pi_2(\hat{R}))$. We claim that for every $\mathbb{K}$-database, we have that $E_{\mathrm{adom}}^\mathbf{I}= \mathrm{adom}(\mathbf{I})$. 

Let $\mathbf{I}=(R)$ be a $\mathbb{K}$-database 
and let $a$ be an element of the domain $D$. First, observe that,  using the definitions of the semantics of $\mathcal{BRA}$-expressions, we have that
$$(\pi_1(\hat{R}) \cup \pi_2(\hat{R}))^\mathbf{I}(a) = 
((\pi_1(\hat{R}))^\mathbf{I}\cup (\pi_2(\hat{R}))^\mathbf{I})(a) =  \sum_{ R(a,b) \neq 0} R(a,b)~~ + \sum_{ R( b,a) \neq 0} R(b,a).$$

 If $a$ is in the active domain $\mathrm{adom}(\mathbf{I})$ of 
$\mathbf{I}$,  there is at least one element $b\in D$ such that $R(a,b) \not = 0$ or $R(b,a) \not = 0$.
Therefore, at least one summand in at least one of the two sums above is different from $0$. Since $\mathbb{K} $ is a zero-sum-free semiring, it follows  that
$\sum_{ R(a,b) \neq 0} R(a,b) ~+~ \sum_{ R( b,a) \neq 0} R(b,a)\not =0$, hence 
$(\pi_1(\hat{R}) \cup \pi_2(\hat{R}))^\mathbf{I}(a)\not = 0$. Consequently, $a\in (\supp(\pi_1(\hat{R}) \cup \pi_2(\hat{R})))^\mathbf{I}= E_{\mathrm{adom}}^\mathbf{I}$.

If $a$ is not in the active domain of $\mathbf{I}$, then each summand in both of the two sums above is equal to $0$, hence $(\pi_1(\hat{R}) \cup \pi_2(\hat{R}))^\mathbf{I}(a)=0$ and so $a \not \in (\supp(\pi_1(\hat{R}) \cup \pi_2(\hat{R})))^\mathbf{I}= E_{\mathrm{adom}}^\mathbf{I}$.
Note that in this step we did not use the hypothesis that $\mathbb{K}$ is a zero-sum-free semiring.
\end{proof}

 \begin{Rmk}
 It is important to point out that if\/ $\mathbb{K}$ is not a zero-sum-free semring, then the $\mathcal{BRA}$-expression $\supp(\pi_1(\hat{R}) \cup \pi_2(\hat{R}))$ does not define the active domain. Indeed, suppose that $k_1$ and $k_2$ are two elements in the universe $K$ of\/ $\mathbb{K}$ such that $k_1+k_2=0$. Let $a,b$ be two distinct elements of the domain $D$ and consider the $\mathbb{K}$-database  $\mathbf{I}=(R)$, where $R(a,b)=k_1$,  $R(b,a)=k_2$, and
 $R(c,d)=0$, for all other pairs $(c,d)$ from $D$. Then $\mathrm{adom}(\mathbf{I})=\{a,b\}$, but 
 $(\pi_1(\hat{R}) \cup \pi_2(\hat{R}))^\mathbf{I}(a)= k_1+k_2=0$
 and $(\pi_1(\hat{R}) \cup \pi_2(\hat{R}))^\mathbf{I}(b)=k_2+k_1=0$, hence $E^\mathbf{I}_{\mathrm{adom}}=\emptyset \not = \mathrm{adom}(\mathbf{I})$.
 \end{Rmk}
 
\begin{Def} [Query] Let $n\geq 0$. 
A $n$-ary \emph{query} $q$    is a mapping defined on $\mathbb{K}$-databases and such that on every $\mathbb{K}$-database $\mathbf{I}$,
$q$ returns an $n$-ary $\mathbb{K}$-relation  $q^\mathbf{I}$ on $\mathrm{adom}(\mathbf{I})$.
% (i.e., for every $n$-tuple ${\bf t}$ with $q^\mathbf{I}({\bf t}) \not = 0$, we have that ${\bf t}\in (\mathrm{adom}(\mathbf{I}))^n$).
\end{Def}

Clearly, every $\mathcal{BRA}$-expression $E$ of arity $n\geq 0$ gives rise to an $n$-ary query associating with every $\mathbb{K}$-database $\mathbf{I}$ the $n$-ary $\mathbb{K}$-relation $E^\mathbf{I}$ as defined in Definition~\ref{algebraop}. 

\begin{Def}
Given a $\mathbb{K}$-database
$\mathbf{I}=(R_1,\ldots,R_p)$, we write  $\mathbf{A}(\mathbf{I})$ to denote the $\mathbb{K}$-structure with universe the active domain of\/ $\mathbf{I}$ and with relations those of\/ $\mathbf{I}$. In symbols, we have  
$\mathbf{A}(\mathbf{I})= (\mathrm{adom}(\mathbf{I}),R_1,\ldots,R_p)$. 
\end{Def}

Clearly, every $\mathcal{BRC}$-formula $\varphi$ with $n\geq 0$ free variables, gives  
to an $n$-ary query associating with every $\mathbb{K}$-database $\mathbf{I}$ the $n$-ary $\mathbb{K}$-relation $\varphi^{\mathbf{A}(\mathbf{I})}$ as defined in Definition~\ref{brc-eval-def}.

\begin{Def}[Domain Independence]
A $\mathcal{BRC}$-formula $\varphi$ is \emph{domain independent} if for every $\mathbb{K}$-database $\mathbf{I}= (R_1,\ldots,R_p)$ and  every   $\mathbb{K}$-structure $\mathbf{B}=(B,R_1,\ldots,R_p)$ with  $\mathrm{adom}(\mathbf{I}) \subseteq B \subseteq D$, we have $\varphi^\mathbf{B}= \varphi^{\mathbf{A}(\mathbf{I})}$.
\end{Def}

The next proposition describes certain properties of domain independent formulas that will be used in the proofs of the main results. 
%In these proofs, we will come across additional closure properties of domain independent formulas.

\begin{Pro}\label{pro:dom-ind}
Let $\tau=(\hat{R_1},\ldots \hat{R_p})$ be a schema. The following statements are true.

\begin{enumerate}
\item Let $\varphi_1(y_1,\ldots,y_n)$ and $\varphi_2(y_1,\ldots,y_n)$ be two $\mathcal{BRC}(\tau)$-formulas with the same free variables. If $\varphi_1$ and $\varphi_2$ are domain independent $\mathcal{BRC}$-formulas, then $\varphi_1 \vee \varphi_2$
and $\varphi_1 \butnot \varphi_2$ are domain independent.

\item If $\varphi_1(y_1,\ldots,y_n)$ and $\varphi_2(y_{n+1},\ldots,y_{n+m})$ are two domain independent $\mathcal{BRC}(\tau)$-formulas,\\ where
$y_1,\ldots,y_{n+m}$ are distinct variables, then
$\varphi_1\wedge \varphi_2$ is domain independent.

\item If $\varphi_1(y_1,\ldots,y_n)$ is a domain independent $\mathcal{BRC}(\tau)$-formula, then $\nabla \varphi$ is domain independent.

\item If $\varphi(y_1,\ldots,y_n)$ is a $\mathcal{BRC}(\tau)$-formula that is domain independent, then $\exists y_i \varphi$ is 
domain independent.

\end{enumerate}

\end{Pro}

\begin{proof}
For the first part, we give the proof for the case of the connective $\vee$. The proof for the case of the connective $\butnot$ is entirely analogous.
Let $\mathbf{I}=(R_1,\ldots,R_p)$ be a $\mathbb{K}$-database and let $\mathbf{B}=(B,R_1,\ldots,R_p)$ be a $\mathbb{K}$-structure with $\mathrm{adom}(\mathbf{I})\subseteq B\subseteq D$.
By the domain independence of $\varphi_1$ and $\varphi_2$, we have that
$\varphi_1^\mathbf{B}= \varphi_1^{\mathbf{A}(\mathbf{I})} $ and $\varphi_2^\mathbf{B}= \varphi_2^{\mathbf{A}(\mathbf{I})} $.
By combining these facts with the second part of Proposition~\ref{pro:basic}, we have that 
$(\varphi_1\vee \varphi_2)^\mathbf{B}=
\varphi_1^\mathbf{B}\cup  \varphi_2^\mathbf{B}= \varphi_1^{\mathbf{A}(\mathbf{I})} \cup  \varphi_2^{\mathbf{A}(\mathbf{I})} = (\varphi_1\vee \varphi_2)^{\mathbf{A}(\mathbf{I})}
$, which establishes the domain independence of the formula $\varphi_1 \vee \varphi_2$. 

For the second part, let $\mathbf{I}=(R_1,\ldots,R_p)$ be a $\mathbb{K}$-database and let $\mathbf{B}=(B,R_1,\ldots,R_p)$ be a $\mathbb{K}$-structure with $\mathrm{adom}(\mathbf{I})\subseteq B\subseteq D$. By the domain independence of $\varphi_1$ and $\varphi_2$, we have that
$\varphi_1^\mathbf{B}= \varphi_1^{\mathbf{A}(\mathbf{I})} $ and $\varphi_2^\mathbf{B}= \varphi_2^{\mathbf{A}(\mathbf{I})} $. By combining these facts with the third part of Proposition~\ref{pro:basic}, we have that $(\varphi_1\wedge \varphi_2)^\mathbf{B}=
\varphi_1^\mathbf{B}\times \varphi_2^\mathbf{B}= 
\varphi_1^{\mathbf{A}(\mathbf{I})}\times \varphi_2^{\mathbf{A}(\mathbf{I})} =(\varphi_1\wedge \varphi_2)^{\mathbf{A}(\mathbf{I})}$. 

For the third part, let $\mathbf{I}=(R_1,\ldots,R_p)$ be a $\mathbb{K}$-database and let $\mathbf{B}=(B,R_1,\ldots,R_p)$ be a $\mathbb{K}$-structure with $\mathrm{adom}(\mathbf{I})\subseteq B\subseteq D$.  By the domain  independence of $\varphi$, we have that $\varphi^{\mathbf{A}(\mathbf{I})} = \varphi^\mathbf{B}$. By the domain  independence of $\varphi$, we have that $\varphi^{\mathbf{A}(\mathbf{I})} = \varphi^\mathbf{B}$. By combining these facts with the fourth part of Proposition~\ref{pro:basic}, we have that $(\nabla \varphi)^\mathbf{B} = \supp (\varphi^\mathbf{B}) = \supp(\varphi^{\mathbf{A}(\mathbf{I})}) = (\nabla \varphi)^{\mathbf{A}(\mathbf{I})}$. 

For the fourth part, let $\mathbf{I}=(R_1,\ldots,R_p)$ be a $\mathbb{K}$-database and let $\mathbf{B}=(B,R_1,\ldots,R_p)$ be a $\mathbb{K}$-structure with $\mathrm{adom}(\mathbf{I})\subseteq B\subseteq D$. By the domain  independence of $\varphi$, we have that $\varphi^{\mathbf{A}(\mathbf{I})} = \varphi^\mathbf{B}$. By combining these facts with the fifth part of Proposition~\ref{pro:basic}, we have that $(\exists y_i \varphi)^\mathbf{B} = \pi_v(\varphi^\mathbf{B}) = \pi_v(\varphi^{\mathbf{A}(\mathbf{I})}) = (\exists y_i \varphi)^{\mathbf{A}(\mathbf{I})}$,
where $v=(i_1,\ldots,i_{n-1})$ is the sequence of distinct indices from $\{1,\ldots,n\}\setminus \{i\}$ ordered by the standard order on $N$. \qedhere
\end{proof}

\begin{Rmk} \label{rmk:dom-ind} It should be pointed out that, in the first part of Proposition \ref{pro:dom-ind}, the assumption that the $\mathcal{BRC}(\tau)$-formulas
$\varphi_1$ and $\varphi_2$ have the same free variables is of the essence. 

To see this, assume that
$\hat{R_1}$ is a unary relation symbol and $\hat{R_2}$ is a binary relation symbol. Let $\varphi(y_1,y_2)$ be the
$\mathcal{BRC}(\tau)$-formula $\hat{R_1}(y_1)\vee \hat{R_2}(y_1,y_2)$, which is the disjunction of the domain independent $\mathcal{BRC}(\tau)$-formulas
$\hat{R_1}(y_1)$ and $\hat{R_2}(y_1,y_2)$.  We claim that $\varphi(y_1,y_2)$ is not domain independent. Let ${\mathbf I}=(R_1,R_2)$ be the $\mathbb K$-database such that $R_1(a) = 1$, $R_1(c) = 0$ if $c\not =a$, $R_2(a,a) =1$, $R_2(c,d)=0$ if $(c,d)\not = (a,a)$. Clearly, $\varphi^{\mathbf{A}(\mathbf{I})}(a,a)=1+1=1$ and $\varphi^{\mathbf{A}(\mathbf{I})}(c,d)=0$ for $(c,d)\not = (a,a)$. Consider the $\mathbf K$-structure
${\mathbf B}=(\{a,b\}, R_1,R_2)$, where $b\not =a$.  Then $\varphi^{\mathbf B}(a,b)= R_1(a) + R_2(a,b)= 1+0 = 1$. Thus,
$\varphi^{\mathbf{A}(\mathbf{I})}\not = \varphi^{\mathbf B}$, hence the formula $\varphi(y_1,y_2)$ is not domain independent.
\end{Rmk}

We are now ready to state and prove the first main theorem connecting basic relational algebra to basic relational calculus.

{\begin{Thm}\label{Codd1} Let\/ $\mathbb{K}$ be a zero-sum-free semiring with monus and support, let $\tau=(\hat{R_1},\ldots \hat{R_p})$ be a schema, and let $q$ be an $n$-ary query. The following
statements are equivalent:
\begin{itemize}
\item[1.] There is a $\mathcal{BRA}(\tau)$-expression $E$ of arity $n$ such that
$q^\mathbf{I}= E^\mathbf{I}$, for every $\mathbb{K}$-database  $\mathbf{I}$.
\item[2.] There is a domain independent $\mathcal{BRC}(\tau)$-formula $\varphi(y_1,\ldots,y_n)$ such
that $q^\mathbf{I}= \varphi^{\mathbf{A}(\mathbf{I})}$, for every $\mathbb{K}$-database  $\mathbf{I}$.
\item[3.] There is a $\mathcal{BRC}(\tau)$-formula
$\varphi(y_1,\ldots,y_n)$ such that $q^\mathbf{I}= \varphi^{\mathbf{A}(\mathbf{I})}$, for every $\mathbb{K}$-database  $\mathbf{I}$.
\end{itemize}
\end{Thm}}

\begin{proof}
The implication  $(2) \implies (3)$ is obvious.
The implications $(1) \implies (2)$  and $(3) \implies (1)$ are established  by induction on the construction of $\mathcal{BRA}(\tau)$-expressions and of $\mathcal{BRC}(\tau)$-formulas, respectively.  In some  steps in the proof of the implication $(3) \implies (1)$, we  will make certain simplifying assumptions about the arity of the relational algebra expressions and the pattern of free variables of the relational calculus formulas. In each such step, the proof extends to the general case with minor modifications only.

\smallskip

\noindent $(1) \implies (2):$  We proceed by induction on the construction of $\mathcal{BRA}(\tau)$-expressions. For every such expression $E$,  we produce a $\mathcal{BRC}(\tau)$-formula $\varphi_E$ such that $E^\mathbf{I}=\varphi^{\mathbf{A}(\mathbf{I})}_E$, for every $\mathbb{K}$-database $\mathbf{I}$.  Moreover, in each case, we show  that the $\mathcal{BRC}(\tau)$-formula $\varphi_E$ is domain independent.
\begin{itemize}
 \item Assume that $E$ is one of the relation symbols $\hat{R_i}$, $1\leq i\leq p$.
In this case, we take  $\varphi_E$ to be the atomic formula $\hat{R_i}(x_1,\ldots,x_n)$, where $n$ is the arity of $\hat{R_i}$. To see why this works, 
let $\mathbf{I}=(R_1,\ldots,R_p)$ be a $\mathbb{K}$-database. From Definition~\ref{algebraop} of the semantics of $\mathcal{BRA}(\tau)$-expressions  and  the first part of Proposition~\ref{pro:basic}, we have  
 $E^\mathbf{I}=R_i= (\hat{R_i}(x_1,\ldots,x_n))^{\mathbf{A}(\mathbf{I})}=\varphi_E^{\mathbf{A}(\mathbf{I})}$, which establishes the correctness of the formula $\varphi_E$. Furthermore, the first part of Proposition~\ref{pro:basic} immmediately implies the domain independence of the formula $\varphi_E$.

\item Assume that $E$ is a
$\mathcal{BRA}(\tau)$-expression of the form
$E_1 \cup E_2$ or of the form $E_1\setminus E_2$, where $E_1$ and $E_2$ are $\mathcal{BRA}(\tau)$-expressions of the same arity, say $n\geq 0$. By induction hypothesis, there are domain independent  $\mathcal{BRC}(\tau)$-formulas $\varphi_1$
and $\varphi_2$ each with $n$ free variables such that for every $\mathbb{K}$-database $\mathbf{I}$, we have that $E_i^\mathbf{I}= \varphi_i^{\mathbf{A}(\mathbf{I})}$, for $i=1,2$. If $E$ is $E_1 \cup E_2$,  we take $\varphi_E$ to be the $\mathcal{BRC}(\tau)$-formula $\varphi_1
\vee \varphi_2$; if $E$ is $E_1\setminus E_2$, we take $\varphi_E$ to be the $\mathcal{BRC}(\tau)$-formula $\varphi_1
\butnot \varphi_2$.
The correctness of the formula $\varphi_E$ follows from the semantics of $\mathcal{BRA}(\tau)$-expressions in Definition~\ref{algebraop}, the inductive hypothesis, and Proposition~\ref{pro:basic}. Specifically, if $\mathbf{I}$ is a $\mathbb{K}$ database, then $E^\mathbf{I}= (E_1\cup E_2)^\mathbf{I}= E_1^\mathbf{I}\cup E_2^\mathbf{I}= 
\varphi_1^{\mathbf{A}(\mathbf{I})} \cup
\varphi_2^{\mathbf{A}(\mathbf{I})}= (\varphi_1 \vee \varphi_2)^{\mathbf{A}(\mathbf{I})}=\varphi_E^{\mathbf{A(I})}$, as desired. The proof for the case of the difference operation $\setminus$ is entirely analogous.
The domain independence of the formula $\varphi_E$ follows from the inductive hypothesis and the first part of Proposition~\ref{pro:dom-ind}.

\item Assume that 
$E$ is a $\mathcal{BRA}(\tau)$-expression of the form $E_1 \times E_2$, where 
$E_1$ is a $\mathcal{BRA}(\tau)$-expression of arity $n$ and $E_2$ is a $\mathcal{BRA}(\tau)$-expression of arity $m$. By induction hypothesis, there are domain independent formulas $\mathcal{BRC}(\tau)$-formulas $\varphi_1$ and $\varphi_2$ such that $\varphi_1$ has $n$ free variables, $\varphi_2$ has $m$ free variables, and for every  for every $\mathbb{K}$-database $\mathbf{I}$, we have that $E_i^\mathbf{I}= \varphi_i^{\mathbf{A}(\mathbf{I})}$, for $i=1,2$. By renaming variables if necessary, we may assume that $\varphi_1$ and $\varphi_2$ have no free variables in common, so they are of the form $\varphi_1(y_1,\ldots,y_n)$ and $\varphi_2(y_{n+1},\ldots,y_{n+m})$, where $y_1,\ldots,y_{n+m}$ are distinct variables. We take $\varphi_E$ to be the $\mathcal{BRC}(\tau)$-formula 
$\varphi_1(y_1,\ldots,y_n) \wedge  \varphi_2(y_{n+1},\ldots,y_{n+m})$, which has $n+m$ free variables.

To show the correctness of the formula $\varphi_E$, we have to show
that if $\mathbf{I}$ is a $\mathbb{K}$-database, then $E^\mathbf{I}=\varphi_E^{\mathbf{A}(\mathbf{I})}$. This holds because $$E^\mathbf{I}= E_1^\mathbf{I}\times E_2^\mathbf{I}= \varphi_1^{\mathbf{A}(\mathbf{I})}\times
\varphi_2^{\mathbf{A}(\mathbf{I})}=
(\varphi_1\wedge \varphi_2)^{\mathbf{A}(\mathbf{I})}= \varphi_E^{\mathbf{A}(\mathbf{I})},$$
where the first equation follows from the semantics of $\mathcal{BRA}(\tau)$-expressions in Definition~\ref{algebraop}, the second follows from the inductive hypothesis, and  third follows from the third part of Proposition~\ref{pro:basic}, and the fourth follows from the definition of $\varphi_E$. The domain independence of the formula $\varphi_E$ follows by the inductive hypothesis and the second part of Proposition \ref{pro:dom-ind}.

\item Assume that $E$ is a $\mathcal{BRA}(\tau)$-expression of the form $\pi_v(E_1)$, where $E_1$ is an $n$-ary $\mathcal{BRA}(\tau)$-expression and $v=(i_1,\ldots,i_k)$ is a sequence of distinct indices from $\{1,\ldots,n\}$.
By induction hypothesis, there is a 
domain independent $\mathcal{BRC}(\tau)$-formula $\varphi_1(y_1,\ldots,y_n)$ such that for every $\mathbb{K}$-database $\mathbf{I}$, we have 
$E_1^\mathbf{I}=\varphi_1^{\mathbf{A}(\mathbf{I})}$. In this case, we take
$\varphi_E$ to be the $\mathcal{BRC}(\tau)$-formula $\exists y_{j_1}\ldots \exists y_{j_{n-k}}\varphi_1$, where $j_1,\ldots,j_{n-k}$ are the indices in $\{1,\ldots,n\}$ that do not appear in the sequence $v=(i_1,\ldots,i_k)$. Note that the free variables of $\varphi_E$ are
$y_{i_1},\ldots,y_{i_k}$.
For example, if
$E$ is of the form $\pi_{1,3}(E_1)$, where $E_1$ has arity $4$, then
$\varphi_E$ is $\exists y_2\exists y_4\varphi_1$ and its free variables are $y_1,y_3$.

To show the correctness of the formula $\varphi_E$, we have to show that if $\mathbf{I}$ is a $\mathbb{K}$-database, then $E^\mathbf{I}=\varphi_E^{\mathbf{A}(\mathbf{I})}$. This holds because 
$$E^\mathbf{I} = \pi_v(E^\mathbf{I}_1) = \pi_v(\varphi_1^\mathbf{\mathbf{A}(\mathbf{I})}) = (\exists y_{j_1}\ldots \exists y_{j_{n-k}}\varphi_1)^\mathbf{\mathbf{A}(\mathbf{I})} =  \varphi_E^\mathbf{\mathbf{A}(\mathbf{I})},$$
where the first equation follows from the semantics of $\mathcal{BRA}(\tau)$-expressions in Definition~\ref{algebraop}, the second follows from the inductive hypothesis, the third follows by iterating the fifth part of Proposition~\ref{pro:basic} to handle the string of quantifiers $\exists y_{j_1}\ldots \exists y_{j_{n-k}}$, and the fourth follows from the definition of $\varphi_E$. The domain independence of the formula $\varphi_E$ follows from the inductive hypothesis and the fourth part of Proposition \ref{pro:dom-ind}.

\item Assume that  $E$ is a $\mathcal{BRA}(\tau)$-expression of the form  $\sigma_\theta(E_1)$, where $\theta$ is a  selection condition and $E_1$ is an $n$-ary $\mathcal{BRA}(\tau)$-expression. By induction hypothesis, there is a 
domain independent $\mathcal{BRC}(\tau)$-formula $\varphi_1(y_1,\ldots,y_n)$ such that for every $\mathbb{K}$-database $\mathbf{I}$, we have 
$E_1^\mathbf{I}=\varphi_1^{\mathbf{A}(\mathbf{I})}$.  We take $\varphi_E$  to be the formula  $\varphi_{1} \wedge \nabla (\varphi_\theta)$, where  $\varphi_1(y_1, \dots, y_n)$ is the formula given by the inductive hypothesis for $E_1$, and 
$\varphi_\theta$ is a $\mathcal{BRC}(\tau)$-formula expressing the selection condition $\theta$ (recall that $\mathcal{BRC}(\tau)$ can express negated equalities, as spelled out in Remark \ref{rmk:neg-equal}).

% *** This is not right. ***

% $\varphi_\theta$ may be more complicated than a conjunction of equalities and negated equalities. Please use the description  of the selection condition in the definition of the selection operation in relational algebra in Section 3.1.

% *******************

% Since equalities and negated equalities take values either $0$ or $1$, a conjunction of such formulas will take value either $0$ or $1$. 

% **** This will no longer be true if we have different selection conditions (e.g., even if we have just disjunctions of equalities we may get values other than 0 or 1.

% I think that we need to use the support of $\varphi_\theta$ here. Please refer again to Section 3.1

% *****************************

To show the correctness of the formula $\varphi_E$, we have to show that if $\mathbf{I}$ is a $\mathbb{K}$-database, then $E^\mathbf{I}=\varphi_E^{\mathbf{A}(\mathbf{I})}$. Let $(a_1,\ldots,a_{n}) \in \mathrm{adom}(\mathbf{I})^n$. It is not difficult to see that $P_\theta(a_1,\ldots,a_{n}) =  (\nabla\varphi_\theta)^{\mathbf{A}(\mathbf{I})}(a_1,\ldots,a_{n})$.
The reason is that,
if $p,q \in \{0,1\}$, then $\mathsf{s}(p+ q)=\begin{cases} 1 & \text{if } p=1 \text{ or } q=1\\ 0 & \text{otherwise} \end{cases}$ and $\mathsf{s}(p\cdot q)=\begin{cases} 1 & \text{if } p=q=1\\ 0 & \text{otherwise} \end{cases}$, hence   disjunction and conjunction in  $(\nabla\varphi_\theta)$ are expressed by means of addition and product prefixed by $\mathsf{s}$. Now, by induction hypothesis, we have that $E_1^\mathbf{I}(a_1,\ldots,a_{n}) =  \varphi_{1}^{\mathbf{A}(\mathbf{I})}(a_1,\ldots,a_{n}) $, hence $ \varphi_{1}^{\mathbf{A}(\mathbf{I})}(a_1,\ldots,a_{n})  \cdot  \varphi_\theta^{\mathbf{A}(\mathbf{I})}(a_1,\ldots,a_{n})  = E_1^\mathbf{I}(a_1,\ldots,a_{n}) \cdot   P_\theta(a_1,\ldots,a_{n})= \sigma_\theta(E_1)^\mathbf{I}(a_1,\ldots,a_{n})$. If $(a_1,\ldots,a_{n})\notin \mathrm{adom}(\mathbf{I})^n$, by induction hypothesis, $E_1^\mathbf{I}(a_1,\ldots,a_{n}) = \varphi_{1}^{\mathbf{A}(\mathbf{I})}(a_1,\ldots,a_{n})= 0 $, and thus $$\varphi_{1}^{\mathbf{A}(\mathbf{I})}(a_1,\ldots,a_{n})  \cdot  \varphi_\theta^{\mathbf{A}(\mathbf{I})}(a_1,\ldots,a_{n})= 0= E_1^\mathbf{I}(a_1,\ldots,a_{n}) \cdot  P_\theta(a_1,\ldots,a_{n})=\sigma_\theta(E_1)^\mathbf{I}(a_1,\ldots,a_{n}).$$

Finally, we show that the formula $\varphi_E$ is domain independent. Let $\mathbf{I}=(R_1,\ldots,R_p)$ be a $\mathbb{K}$-database and let $\mathbf{B}=(B,R_1,\ldots,R_p)$ be a $\mathbb{K}$-structure with $\mathrm{adom}(\mathbf{I})\subseteq B\subseteq D$. First, observe that  we have  $ \varphi_\theta^\mathbf{B}(a_1,\ldots,a_{n}) = \varphi_\theta^{\mathbf{A}(\mathbf{I})}(a_1,\ldots,a_{n}) $. Furthermore, by induction hypothesis, we have that $\varphi_{1}^\mathbf{B}(a_1,\ldots,a_{n}) = \varphi_{1}^{\mathbf{A}(\mathbf{I})}(a_1,\ldots,a_{n})$.  Thus, $$\varphi_{1}^\mathbf{B}(a_1,\ldots,a_{n}) \cdot \varphi_\theta^\mathbf{B}(a_1,\ldots,a_{n}) = \varphi_{1}^{\mathbf{A}(\mathbf{I})}(a_1,\ldots,a_{n})\cdot \varphi_\theta^{\mathbf{A}(\mathbf{I})}(a_1,\ldots,a_{n}),$$
which entails that  $\varphi_E^\mathbf{B}(a_1,\ldots,a_{n}) = \varphi_E^{\mathbf{A}(\mathbf{I})}(a_1,\ldots,a_{n})$, as desired. If $(a_1,\ldots,a_{n}) $ is not a tuple from $\mathrm{adom}(\mathbf{I})$, since by induction hypothesis  $\varphi_{1}^\mathbf{B}(a_1,\ldots,a_{n}) = \varphi_{1}^{\mathbf{A}(\mathbf{I})}(a_1,\ldots,a_{n})$ and $\varphi_{1}^{\mathbf{A}(\mathbf{I})}(a_1,\ldots,a_{n}) =0$, we must have that  $\varphi_E^\mathbf{B}(a_1,\ldots,a_{n}) = 0=\varphi_E^{\mathbf{A}(\mathbf{I})}(a_1,\ldots,a_{n})$, as desired.

\item Assume that $E$ is a $\mathcal{BRA}(\tau)$-expression of the form $\text{supp}(E_1)$, where $E_1$ is an $n$-ary $\mathcal{BRA}(\tau)$-expression. By induction hypothesis, there is a 
domain independent $\mathcal{BRC}(\tau)$-formula $\varphi_1(y_1,\ldots,y_n)$ such that for every $\mathbb{K}$-database $\mathbf{I}$, we have 
$E_1^\mathbf{I}=\varphi_1^{\mathbf{A}(\mathbf{I})}$. We take
$\varphi_E$ to be the $\mathcal{BRC}(\tau)$-formula $\nabla \varphi_{1}$.

To show the correctness of the formula $\varphi_E$, we have to show that if $\mathbf{I}$ is a $\mathbb{K}$-database, then $E^\mathbf{I}=\varphi_E^{\mathbf{A}(\mathbf{I})}$. This holds because 
$$E^\mathbf{I} = \supp(E^\mathbf{I}_1) = \supp(\varphi_1^\mathbf{\mathbf{A}(\mathbf{I})}) = (\nabla \varphi_{1})^\mathbf{\mathbf{A}(\mathbf{I})} =  \varphi_E^\mathbf{\mathbf{A}(\mathbf{I})},$$
where the first equation follows from the semantics of $\mathcal{BRA}(\tau)$-expressions in Definition~\ref{algebraop}, the second follows from the inductive hypothesis,  the third follows from the fourth part of Proposition~\ref{pro:basic}, and the fourth follows from the definition of $\varphi_E$. The domain independence of the formula $\varphi_E$ follows from the inductive hypothesis and the fourth part of Proposition \ref{pro:dom-ind}.
\end{itemize}

$(3) \implies (1):$ We proceed by induction on the size of the $\mathcal{BRC}(\tau)$-formula $\varphi$.
\begin{itemize}
\item Assume that $\varphi (y_1, y_2)$ is a $\mathcal{BRC}(\tau)$-formula of the form $y_1= y_2$. In this case, we take $E_\varphi$ to be $$\text{supp}(\sigma_\theta(E_{\mathrm{adom}} \times E_{\mathrm{adom}} )),$$ where $\theta$ is $x_1=x_2$ and thus 

$P_\theta (a,b)=\begin{cases} 1 & \text{if } a=b \\ 0 & \text{otherwise}.\end{cases}$

Observe that  for $a,b$ in  $\mathrm{adom}(\mathbf{I})$, if $a=b$, we have that $E_{\mathrm{adom}}^\mathbf{I}(a) \neq 0$, $E_{\mathrm{adom}}^\mathbf{I}(b) \neq 0$, and thus $E_{\mathrm{adom}}^\mathbf{I}(a) \cdot E_{\mathrm{adom}}^\mathbf{I}(b) \neq 0$. Furthermore, since $P_\theta (a,b)=1$, we have that $$(\sigma_\theta(E_{\mathrm{adom}} \times E_{\mathrm{adom}}))^\mathbf{I}(a,b) = (E_{\mathrm{adom}} \times E_{\mathrm{adom}} )^\mathbf{I}(a,b),$$
which means that
$$(\sigma_\theta(E_{\mathrm{adom}} \times E_{\mathrm{adom}}))^\mathbf{I}(a,b) \neq 0,$$
hence $E_\varphi^\mathbf{I}(a,b) =1$. Similarly, if $a\neq b$, we have that $E_\varphi^\mathbf{I}(a,b) =0$.  Thus, if $a,b \in \mathrm{adom}(\mathbf{I})$, we have that $E_\varphi^\mathbf{I}(a,b) = (y_1= y_2)^{\mathbf{A}(\mathbf{I})}(a,b)$. Now, if $a \not \in \mathrm{adom}(\mathbf{I})$, then  $(y_1= y_2)^{\mathbf{A}(\mathbf{I})}(a,b)=0$ by definition and  $E_{\mathrm{adom}}^\mathbf{I}(a) = 0$, which entails  that $E_\varphi^\mathbf{I}(a,b)= 0$. The case in which  $b \not \in  \mathrm{adom}(\mathbf{I})$ is identical.

\item Assume that $\varphi (y_1, \ldots,y_n)$ is  a $\mathcal{BRC}(\tau)$-formula of the form $\hat{R_i}(y_1, \ldots,y_n)$, where $\hat{R_i}$ is one of the relation symbols of $\tau$. In this case, we take $E_\varphi$ to be simply the relation symbol $\hat{R_i}$. The correctness of $E_\varphi$ follows from Definitions \ref{algebraop}, \ref{FOL-sem-def}, and \ref{brc-eval-def}.

\item
 Assume that $\varphi (y_1, y_2, y_3, y_4)$ is a $\mathcal{BRC}(\tau)$-formula of the form is $\psi(y_1,  y_2, y_4) \butnot \chi(y_2, y_3)$. In this case, we take $E_\varphi$ to be $$\pi_{1,2,4,3}(E_\psi \times E_{\mathrm{adom}}) \setminus \pi_{1,3,4,2}( E_{\mathrm{adom}} \times  E_{\mathrm{adom}} \times E_\chi ). $$
 %where $E_{\mathrm{adom}} = \text{supp}(E_{\mathrm{adom}})$.

For all elements $a_1,a_2,a_3, a_4  \in D$, we have that $E_\varphi^\mathbf{I}(a_1,a_2, a_3, a_4) =$\\ $=(\pi_{1,2,4,3}(E_\psi \times E_{\mathrm{adom}}))^\mathbf{I}(a_1,a_2, a_3, a_4) - (\pi_{1,3,4,2}( E_{\mathrm{adom}} \times  E_{\mathrm{adom}} \times E_\chi ))^\mathbf{I}(a_1,a_2, a_3, a_4).$

If $a_1,a_2,a_3, a_4 \in \mathrm{adom}(\mathbf{I})$, then
 $$(\pi_{1,2,4,3}(E_\psi \times E_{\mathrm{adom}}))^\mathbf{I}(a_1,a_2, a_3, a_4) =\sum_{{\bf w}[1,2,4,3]=a_1,a_2, a_3, a_4~ \text{and}~(E_\psi \times E_{\mathrm{adom}})^\mathbf{I}({\bf w}) \neq 0} (E_\psi \times E_{\mathrm{adom}})^\mathbf{I}({\bf w})$$
 and then where ${\bf w}[1,2,4,3]=a_1,a_2, a_3, a_4$ (and there is exactly one such ${\bf w}$, namely $(a_1, a_2, a_4, a_3)$), $$(\pi_{1,2,4,3}(E_\psi \times E_{\mathrm{adom}}))^\mathbf{I}(a_1,a_2, a_3, a_4)=(E_\psi \times E_{\mathrm{adom}})^\mathbf{I}({\bf w}) =
E_\psi^\mathbf{I}(a_1,a_2, a_4) \cdot E_{\mathrm{adom}}^\mathbf{I}(a_3)  $$
and $E_\psi^\mathbf{I}(a_1,a_2, a_4) \cdot E_{\mathrm{adom}}(a_3) = E_\psi^\mathbf{I}(a_1,a_2, a_4) \cdot 1 = \psi^{\mathbf{A}(\mathbf{I})}(a_1,a_2, a_4),$\\
 where the last equality follows from the  induction hypothesis, so
 $$(\pi_{1,2,4,3}(E_\psi \times E_{\mathrm{adom}}))^\mathbf{I}(a_1,a_2, a_3, a_4) = \psi^{\mathbf{A}(\mathbf{I})}(a_1,a_2, a_4).$$
 Similarly, $(\pi_{1,3,4,2}( E_{\mathrm{adom}} \times  E_{\mathrm{adom}} \times E_\chi ))^\mathbf{I}(a_1,a_2, a_3, a_4)= \chi^{\mathbf{A}(\mathbf{I})}(a_2,a_3).$
 Thus, when\\$a_1,a_2,a_3, a_4 \in \mathrm{adom}(\mathbf{I})$, we have that
$$E_\varphi^\mathbf{I}(a_1,a_2, a_3, a_4)  =  \psi^{\mathbf{A}(\mathbf{I})}(a_1,a_2, a_4) - \chi^{\mathbf{A}(\mathbf{I})}(a_2,a_3) =\varphi^{\mathbf{A}(\mathbf{I})}(a_1,a_2, a_3, a_4),$$
as desired. If at least one of $a_1,a_2,a_3, a_4$ is not in $\mathrm{adom}(\mathbf{I})$, then, by Definition \ref{brc-eval-def}, we have that $\psi^{\mathbf{A}(\mathbf{I})}(a_1,a_2, a_4)=0 $ or $ \chi^{\mathbf{A}(\mathbf{I})}(a_2,a_3)=0$; hence, by induction hypothesis,  $E_\psi^\mathbf{I}(a_1,a_2, a_4)=0 $ or $0 = E_\chi^\mathbf{I}(a_2, a_3)$. Suppose then that $a_1 \notin \mathrm{adom}(\mathbf{I})$, but
$a_2,a_3 \in \mathrm{adom}(\mathbf{I})$ and
$\chi^{\mathbf{A}(\mathbf{I})}(a_2,a_3)\neq 0 $ (the other cases are similar). Thus, 
$\psi^{\mathbf{A}(\mathbf{I})}(a_1,a_2, a_4)=0 = E_\psi^\mathbf{I}(a_1,a_2, a_4) $, but  $ E_\chi^\mathbf{I}(a_2,a_3)=\chi^{\mathbf{A}(\mathbf{I})}(a_2,a_3)\neq 0 $. Now, 
 $$E_{\mathrm{adom}}^\mathbf{I}(a_1) \cdot  E_{\mathrm{adom}}^\mathbf{I}(a_4) \cdot E_\chi^\mathbf{I}(a_2, a_3) = 0 \cdot E_{\mathrm{adom}}^\mathbf{I}(a_4)  \cdot E_\chi^\mathbf{I}(a_2, a_3)= 0,$$
 and $E_\psi^\mathbf{I}(a_1,a_2, a_4) \cdot E_{\mathrm{adom}}^\mathbf{I}(a_3) = 0 \cdot E_{\mathrm{adom}}^\mathbf{I}(a_3)= 0.$ Thus,
$$(\pi_{1,2,4,3}(E_\psi \times E_{\mathrm{adom}}))^\mathbf{I}(a_1,a_2, a_3, a_4) = 0 = (\pi_{1,3,4,2}( E_{\mathrm{adom}} \times  E_{\mathrm{adom}} \times E_\chi ))^\mathbf{I}(a_1,a_2, a_3, a_4).$$

Then, $E_\varphi(a_1,a_2, a_3, a_4) = 0$ since $0-0=0$, and 
by Definition \ref{brc-eval-def},  $\varphi^{\mathbf{A}(\mathbf{I})}(a_1,a_2,a_3, a_4) = 0$. Thus, $E_\varphi^\mathbf{I}(a_1,a_2, a_3, a_4) = \varphi^{\mathbf{A}(\mathbf{I})}(a_1,a_2,a_3, a_4)$, as desired.

 Note that this inductive step is the first place where the formula $\varphi$ may have no free variables. We deal with this case by producing an expression $E_\varphi$ that defines a $0$-ary query. We argue along the same lines we just  did above, but this time all we have to show is that $\varphi^{\mathbf{A}(\mathbf{I})}(< >) = E_\varphi^\mathbf{I}(< >)$. The details are left to the reader. This same comment applies to the remaining  inductive steps when the formula $\varphi$ has no free variables.
\item 
Assume that $\varphi (y_1, y_2, y_3, y_4)$ is a $\mathcal{BRC}(\tau)$-formula of the form $\psi(y_1,  y_2, y_4) \wedge \chi(y_2, y_3)$. Without loss of generality, assume that the variable $y_5$ does not occurr in the formulas $\psi$ and $\chi$. Let 
$\chi'$ be the formula obtained from $\chi$ by renaming every occurrence of the variable $y_2$ in $\chi$ $y_5$, thus the free variables of $\chi'$ are $y_5$ and $y_3$.
We take $E_\varphi$ to be $$\pi_{1,2,5,3}(\sigma_{\theta}(E_\psi \times E_{\chi'})), $$
where $\theta$ is the selection condition $y_2=y_5$. 

For all elements $a_1,a_2,a_3, a_4 \in D$, we have that  $$E_\varphi^\mathbf{I}(a_1,a_2, a_3, a_4)  = \sum_{{\bf w}[1,2,5,3]=a_1,a_2, a_3, a_4~ \text{and}~\sigma_{\theta}(E_\psi \times E_{\chi'})^\mathbf{I}({\bf w}) \neq 0} (\sigma_{\theta}(E_\psi \times E_{\chi'}))^\mathbf{I}({\bf w}), $$
and that 
$$\sum_{{\bf w}[1,2,5,3]=a_1,a_2, a_3, a_4~ \text{and}~\sigma_{\theta}(E_\psi \times E_{\chi'})^\mathbf{I}({\bf w}) \neq 0} (\sigma_{\theta}(E_\psi \times E_{\chi'}))^\mathbf{I}({\bf w}) $$ is identical to $$ \sum_{{\bf w}[1,2,5,3]=a_1,a_2, a_3, a_4~ \text{and}~\sigma_{\theta}(E_\psi \times E_{\chi'})^\mathbf{I}({\bf w}) \neq 0} (E_\psi^\mathbf{I}({\bf w}[v_1]) \cdot E_{\chi'}^\mathbf{I}({\bf w}[v_2]) \cdot P_\theta({\bf w})),$$
where $v_1=(1,2,4)$ and $v_2=(5,3)$.

If $a_1,a_2,a_3, a_4 \in \mathrm{adom}(\mathbf{I})$, then,  by  induction hypothesis, we have that 
$$\sum_{{\bf w}[1,2,5,3]=a_1,a_2, a_3, a_4~ \text{and}~\sigma_{\theta}(E_\psi \times E_{\chi'})^\mathbf{I}({\bf w}) \neq 0} (E_\psi^\mathbf{I}({\bf w}[v_1]) \cdot E_{\chi'}^\mathbf{I}({\bf w}[v_2]) \cdot P_\theta({\bf w}))$$
is identical to
$$ \sum_{{\bf w}[1,2,5,3]=a_1,a_2, a_3, a_4~ \text{and}~\sigma_{\theta}(E_\psi \times E_{\chi'})^\mathbf{I}({\bf w}) \neq 0} (\psi^{\mathbf{A}(\mathbf{I})} ({\bf w}[v_1]) \cdot \chi'^{\mathbf{A}(\mathbf{I})} ({\bf w}[v_2]) \cdot P_\theta({\bf w})),$$
which in turn is
$$ \psi^{\mathbf{A}(\mathbf{I})} (a_1,a_2,a_4) \cdot \chi'^{\mathbf{A}(\mathbf{I})} (a_2, a_3) \cdot 1 = \psi^{\mathbf{A}(\mathbf{I})} (a_1,a_2,a_4) \cdot \chi'^{\mathbf{A}(\mathbf{I})} (a_2, a_3),$$
since, when ${\bf w}[1,2,5,3]=a_1,a_2, a_3, a_4$, we have that $P_\theta({\bf w}) = 1$   only when ${\bf w}$ is $(a_1,a_2, a_3, a_4,a_2)$ and $P_\theta({\bf w}) = 0$ for all other possible tuples. Thus, $$E_\varphi^\mathbf{I}(a_1,a_2, a_3, a_4)  = \psi^{\mathbf{A}(\mathbf{I})}(a_1,a_2, a_4) \wedge \chi^{\mathbf{A}(\mathbf{I})}(a_2, a_3),$$
as desired.

If at least one of $a_1,a_2,a_3, a_4$ is not in  $\mathrm{adom}(\mathbf{I})$, then, by Definition \ref{brc-eval-def},   $\psi^{\mathbf{A}(\mathbf{I})}(a_1,a_2, a_4)=0 $ or $ \chi^{\mathbf{A}(\mathbf{I})}(a_2,a_3)=0$; hence,   by induction hypothesis, we have that  $E_\psi^\mathbf{I}(a_1,a_2, a_4)=0 $ or $E_\chi^\mathbf{I}(a_2, a_3) = 0$. Suppose then that $a_1 \notin \mathrm{adom}(\mathbf{I})$, but $a_2,a_3 \in \mathrm{adom}(\mathbf{I})$ and $\chi^{\mathbf{A}(\mathbf{I})}(a_2,a_3)\neq 0$ (the other cases are similar).
Thus, 
$\psi^{\mathbf{A}(\mathbf{I})}(a_1,a_2, a_4)=0 = E_\psi^\mathbf{I}(a_1,a_2, a_4)$, but   $ E_\chi^\mathbf{I}(a_2,a_3)=\chi^{\mathbf{A}(\mathbf{I})}(a_2,a_3)\neq 0$.  We have then that
$$E_\psi^\mathbf{I}(a_1,a_2,a_4) \cdot E_{\chi'}^\mathbf{I}(a_2, a_3) \cdot P_\theta(a_1,a_2, a_3, a_4,a_2) = 0.$$
  Thus,
$$(\pi_{1,2,5,3}(\sigma_{\theta}(E_\psi \times E_{\chi'})))^\mathbf{I}(a_1,a_2, a_3, a_4) = 0.$$

Then, $E_\varphi^\mathbf{I}(a_1,a_2, a_3, a_4) = 0$, and 
by Definition \ref{brc-eval-def},  we have that $\varphi^{\mathbf{A}(\mathbf{I})}(a_1,a_2,a_3, a_4) = 0$. Thus, $E_\varphi(a_1,a_2, a_3, a_4) = \varphi ^{\mathbf{A}(\mathbf{I})}(a_1,a_2,a_3, a_4)$, as desired.

\item 
Assume that   $\varphi (y_1, y_2, y_3, y_4)$ is a $\mathcal{BRC}(\tau)$-formula of the form $\psi(y_1,  y_2, y_4) \vee \chi(y_2, y_3)$. In this case, we take $E_\varphi$ to be $$\pi_{1,2,4,3}(E_\psi \times E_{\mathrm{adom}}) \cup \pi_{1,3,4,2}( E_{\mathrm{adom}} \times  E_{\mathrm{adom}} \times E_\chi ). $$
% where $E_{\mathrm{adom}} = \text{supp}(E_{\mathrm{adom}})$.
 
 For all  elements $a_1,a_2,a_3, a_4 \in D$, we have that {\small $$E_\varphi^\mathbf{I}(a_1,a_2, a_3, a_4) = (\pi_{1,2,4,3}(E_\psi \times E_{\mathrm{adom}}))^\mathbf{I}(a_1,a_2, a_3, a_4) + (\pi_{1,3,4,2}( E_{\mathrm{adom}} \times  E_{\mathrm{adom}} \times E_\chi ))^\mathbf{I}(a_1,a_2, a_3, a_4).$$}
If $a_1,a_2,a_3, a_4 \in \mathrm{adom}(\mathbf{I})$, then
 $$(\pi_{1,2,4,3}(E_\psi \times E_{\mathrm{adom}}))^\mathbf{I}(a_1,a_2, a_3, a_4) = \sum_{{\bf w}[1,2,4,3]=(a_1,a_2, a_3, a_4)~ \text{and}~(E_\psi \times E_{\mathrm{adom}})^\mathbf{I}({\bf w}) \neq 0} (E_\psi \times E_{\mathrm{adom}})^\mathbf{I}({\bf w}).$$ 
 There is only one ${\bf w}$ such that 
  ${\bf w}[1,2,4,3]=(a_1,a_2, a_3, a_4)$,  namely ${\bf w}=(a_1, a_2, a_4, a_3)$). Thus,
  $$(\pi_{1,2,4,3}(E_\psi \times E_{\mathrm{adom}}))^\mathbf{I}(a_1,a_2, a_3, a_4)=(E_\psi \times E_{\mathrm{adom}})(a_1,a_2,a_4,a_3) =
E_\psi^\mathbf{I}(a_1,a_2, a_4) \cdot E_{\mathrm{adom}}^\mathbf{I}(a_3)  $$
and
$$ E_\psi^\mathbf{I}(a_1,a_2, a_4) \cdot E_{\mathrm{adom}}^\mathbf{I}(a_3) = E_\psi(a_1,a_2, a_4) \cdot 1 = \psi^{\mathbf{A}(\mathbf{I})}(a_1,a_2, a_4),$$
 where the last equality follows by induction hypothesis, so
 $$(\pi_{1,2,4,3}(E_\psi \times E_{\mathrm{adom}}))^\mathbf{I}(a_1,a_2, a_3, a_4) = \psi^{\mathbf{A}(\mathbf{I})}(a_1,a_2, a_4).$$
 Similarly, 
 $$ (\pi_{1,3,4,2}( E_{\mathrm{adom}} \times  E_{\mathrm{adom}} \times E_\chi ))^\mathbf{I}(a_1,a_2, a_3, a_4)= \chi^{\mathbf{A}(\mathbf{I})}(a_2,a_3).$$
 Thus, when $a_1,a_2,a_3, a_4 \in \mathrm{adom}(\mathbf{I})$, we have
$$E_\varphi^\mathbf{I}(a_1,a_2, a_3, a_4)  =  \psi^{\mathbf{A}(\mathbf{I})}(a_1,a_2, a_4) + \chi^{\mathbf{A}(\mathbf{I})}(a_2,a_3) =\varphi^{\mathbf{A}(\mathbf{I})}(a_1,a_2, a_3, a_4),$$
as desired. 

The argument for the  case in which at least one of 
$a_1,a_2,a_3,a_4$ is not in $\mathrm{adom}(\mathbf{I})$ is similar to the argument for the analogous case in the previous two inductive steps. The details are left to the reader.

%If $a_1,a_2,a_3, a_4 \notin \mathrm{adom}(\mathbf{I})$, then either $\psi^{\mathbf{A}(\mathbf{I})}(a_1,a_2, a_4)=0 $ or $ \chi^{\mathbf{A}(\mathbf{I})}(a_2,a_3)=0$ by definition and, by induction hypothesis, we have that  $E_\psi^\mathbf{I}(a_1,a_2, a_4)=0 $ or $0 = E_\chi^\mathbf{I}(a_2, a_3)$. Suppose then that $a_1 \notin \mathrm{adom}(\mathbf{I})$, $\psi^{\mathbf{A}(\mathbf{I})}(a_1,a_2, a_4)=0 $ and thus $E_\psi^\mathbf{I}(a_1,a_2, a_4)=0 $ (the other cases are similar) but where $ E_\chi^\mathbf{I}(a_2,a_3)=\chi^{\mathbf{A}(\mathbf{I})}(a_2,a_3)\neq 0 $ since $a_2,a_3 \in \mathrm{adom}(\mathbf{I})$. Now, 
 %$$E_{\mathrm{adom}}^\mathbf{I}(a_1) \cdot  E_{\mathrm{adom}}^\mathbf{I}(a_4) \cdot E_\chi^\mathbf{I}(a_2, a_3) = 0 \cdot E_{\mathrm{adom}}^\mathbf{I}(a_4)  \cdot E_\chi^\mathbf{I}(a_2, a_3)= 0,$$
 %and 
 %$$E_\psi^\mathbf{I}(a_1,a_2, a_4) \cdot E_{\mathrm{adom}}^\mathbf{I}(a_3) = 0 \cdot E_{\mathrm{adom}}(a_3)= 0.$$
 % Thus,
%$$(\pi_{1,2,4,3}(E_\psi \times E_{\mathrm{adom}}))^\mathbf{I}(a_1,a_2, a_3, a_4) = 0 = (\pi_{1,3,4,2}( E_{\mathrm{adom}} \times  E_{\mathrm{adom}} \times E_\chi ))^\mathbf{I}(a_1,a_2, a_3, a_4).$$
 
%Then, $E_\varphi^\mathbf{I}(a_1,a_2, a_3, a_4) = 0$, and 
%by definition,  $\varphi^{\mathbf{A}(\mathbf{I})}(a_1,a_2,a_3, a_4) = 0$. Thus, $E_\varphi^\mathbf{I}(a_1,a_2, a_3, a_4) = \varphi^{\mathbf{A}(\mathbf{I})}(a_1,a_2,a_3, a_4)$ as desired.

\item Assume that $\varphi$ is a $\mathcal{BRC}(\tau)$-formula of the form $\nabla \psi$. In this case, we take $E_{\varphi}$ to be $\text{supp}(E_{\psi})$.  By the fourth part of Proposition \ref{pro:basic}, we have that 
$(\nabla \psi)^{\mathbf{A}(\mathbf{I})}  = 
\text{supp}(\psi^{\mathbf{A}(\mathbf{I})})$. By induction hypothesis, we have that $\psi^{\mathbf{A}(\mathbf{I})}= E_\psi^{\mathbf{I}}$, hence 
$(\nabla \psi)^{\mathbf{A}(\mathbf{I})}= \text{supp}(E_\psi^{\bf I})= E_\varphi^{\bf I}$, as desired.

\item Assume that $\varphi$ is a 
$\mathcal{BRC}(\tau)$-formula of the form 
$\exists y_i\psi$, where  $\psi$ is
a $\mathcal{BRC}(\tau)$-formula with $y_1,\ldots,y_n$ as its free variables and
$i$ is an index such that $1\leq i\leq n$.
In this case, we take $E_\varphi$ to be
$\pi_v(E_\psi)$, where
 $v=(i_1,\ldots,i_{n-1})$ is the sequence of distinct indices from $\{1,\ldots,n\}\setminus \{i\}$ ordered by the standard  ordering of the natural numbers.  
 
 By the fifth part of Proposition \ref{pro:basic}, we have that
 $(\exists y_i \psi)^{\mathbf{A}(\mathbf{I})}  = \pi_v(\psi^{\mathbf{A}(\mathbf{I})} )$. By induction hypothesis, we have that $\psi^{\mathbf{A}(\mathbf{I})}= E_\psi^{\mathbf{I}}$, hence
 $(\exists y_i \psi)^{\mathbf{A}(\mathbf{I})}  = \pi_v(E_\psi^{\mathbf{I}} )= (\pi_v(E_\psi))^\mathbf{I}= E_\varphi^\mathbf{I}$, as desired. \qedhere
\end{itemize}
\end{proof}

%The proof of the following result can be found in the Appendix:

%\begin{Pro} \label{pro:zero-sum}
%Let\/ $\mathbb{K}$ be a semiring with monus and support. Then, $\mathbb{K}$ is zero-sum-free if and only if the equivalences in Codd's Theorem hold and for every $\mathbf{I}$, we have that $\mathbf{A} (\mathbf{I})$ is a structure.
%\end{Pro}

\section{Codd's Theorem for Relational Algebra with Division}
\subsection{Relational algebra with  division}

We now introduce an additional relational algebra operation that is definable from the basic ones in the Boolean case but not in the general semiring framework, as we will soon see.
\begin{Def} [Division] \label{def:div} Let $R_1 \colon D^{n_1} \longrightarrow K$ and $R_2 \colon D^{n_2} \longrightarrow K$ be two $\mathbb{K}$-relations with $n_1 >  n_2$. The \emph{division} of  $R_1$ and $R_2$ is the $\mathbb{K}$-relation 
$R_1 \div R_2 \colon D^{n_1 - n_2} \longrightarrow K$  defined by
 $$ R_1 \div R_2({\bf a}) =  
 \mathsf{s}(\sum_{{\bf b} \in D^{n_2}}R_1({\bf a},{\bf b}))~
 \cdot
 \prod_{{\bf b}~:~ R_2({\bf b})\not = 0}  R_1({\bf a}, {\bf b}). $$

% $${\color{red}{???: R_1 \div R_2({\bf t}) =  \prod_{{\bf t}'= ({\bf t}, {\bf s}) \ \text{and}  \ R_2({\bf s}) \neq 0} (R_1({\bf t}') \cdot  R_2({\bf s}))}}$$
We stipulate that the empty product is equal to $1$. Observe that $ R_1 \div R_2 =  R_1 \div \text{supp}( R_2)$.
\end{Def}

Note that the quotient $R_1\div R_2 $ is well-defined. The reason is that, for every tuple ${\bf a} \in D^{n_1-n_2}$, we have that the  product $ \prod_{ R_2({\bf b})\not = 0}  R_1({\bf a}, {\bf b}) $
is finite because  $R_2$ is a $\mathbb K$-relation, thus there are finitely many tuples $\bf b$ such that $R_2({\bf b})\not = 0$. Furthermore, $R_1\div R_2$ has finite support because so does $R_1$, hence we have that 
$\mathsf{s}(\sum_{{\bf b} \in D^{n_2}}R_1({\bf a},{\bf b}))\not = 0$ for only finitely many tuples ${\bf a}\in D^{n_1-n_2}$. Thus, $R_1\div R_2$ is indeed a $\mathbb K$-relation. 
Note also that if $R_2$ is the $n_2$-ary  $\mathbb K$-relation with empty support, then
$ \prod_{ R_2({\bf b})\not = 0}  R_1({\bf a}, {\bf b}) = 1$, for every tuple ${\bf a} \in D^{n_1-n_2}$, because this product is empty. This is why the expression 
$\mathsf{s}(\sum_{{\bf b} \in D^{n_2}}R_1({\bf a},{\bf b}))$ is needed in the definition of $R_1 \div R_2$, as without it we would not have a $\mathbb K$-relation, i.e., a function with finite support.

We  extend the Basic Relational Algebra $\mathcal{BRA}$ to obtain the \emph{Relational Algebra} $\mathcal{RA}$ by adding the following clause to Definition~\ref{bradef}:

\begin{itemize}

\item If $E_1, E_2$ are $\mathcal{RA}(\tau)$-expressions of arities $n_1, n_2$, respectively, such that $n_1 >  n_2$,  then $E_1 \div E_2 $ is an $\mathcal{RA}(\tau)$-expression of arity $n_1-n_2$.
\end{itemize}

We also extend the semantics by adding the following clause to Definition~\ref{algebraop}:
\begin{itemize}
\item $(E_1\div E_2)^{\mathbf I}= E_1^{\mathbf I} \div E_2^{\mathbf I}
$.
\end{itemize}

\begin{Exm}\label{exam1:div}
   If\/ ${\mathbb B}=(\{0,1\}, \wedge,\vee, 0, 1)$ is the Boolean semiring, then
 the division operation introduced here coincides with the division (also known as \emph{quotient}) operation of relational algebra, i.e., if $R_1$ and $R_2$ are ordinary relations with $\textrm {arity}(R_2)<  \textrm {arity}(R_1)$, then
$R_1 \div R_2({\bf a}) = 1$ if and only if for all $\bf b$ with $R_2({\bf b})=1$, we have that $R_1({\bf a}, {\bf b})=1$. 
\end{Exm}

\begin{Exm} \label{exam2-div}
We illustrate the meaning of the division operation for other semirings of interest.

Let\/ ${\mathbb K}=(K,+_{\mathbb K},\cdot_{\mathbb K}, 0_{\mathbb K}, 1_{\mathbb K})$ be a positive semiring with monus and support. Assume that $R_1$ and $R_2$ are two $\mathbb K$-relations with  $\textrm {arity}(R_2)<  \textrm {arity}(R_1)$,
and let ${\bf a}$ be a tuple in $D^{n_1-n_2}$.

\begin{enumerate}
\item If ${\bf a} \not \in \supp(\pi_{1\ldots n_1 -n_2}(R_1))$, then
$R_1\div R_2({\bf a})=0_{\mathbb K}$.

\item If ${\bf a}  \in \supp(\pi_{1\ldots n_1 -n_2}(R_1))$ and there is a tuple ${\bf b}\in \supp(R_2)$ such that
$({\bf a}, {\bf b}) \not \in \supp(R_1)$, then $R_1\div R_2({\bf a})=0_{\mathbb K}$.

\item  If ${\bf a}  \in \supp(\pi_{1\ldots n_1 -n_2}(R_1))$ and for every tuple ${\bf b}\in \supp(R_2)$ we have that
$({\bf a}, {\bf b})  \in \supp(R_1)$, then the following hold:
\begin{enumerate}
\item If\/ $\mathbb K$ is the bag semiring ${\mathbb N}=(N,+,\cdot,0,1)$, then
$R_1\div R_2({\bf a}) = \prod_{ R_2({\bf b})\not = 0}  R_1({\bf a}, {\bf b}).$
\item If\/ $\mathbb K$ is the tropical semiring $\mathbb{T}(N) = (N\cup\{\infty\},\min,+,\infty,0)$ of the natural numbers, then $R_1\div R_2({\bf a}) = \sum_{ R_2({\bf b})\not = 0}  R_1({\bf a}, {\bf b}).$
\item If\/ $\mathbb K$ is the fuzzy semiring ${\mathbb F} =([0,1], \max, \min, 0, 1)$, then $R_1\div R_2({\bf a}) = \min_{ R_2({\bf b})\not = 0}  R_1({\bf a}, {\bf b}).$
\end{enumerate}

\end{enumerate}
    \end{Exm}

In the case of the Boolean semiring
$\mathbb B$, it is well known that the division operation is expressible in terms of the other relational algebra operations; in fact, it is expressible using projection, Cartesian product, and difference (e.g., see~\cite[pages 153-154]{DBLP:books/cs/Ullman80}).  In what follows, we show that the state of affairs is quite different for the bag semiring $\mathbb N$.

\begin{Thm} \label{thm:division-bag}
Let $\mathbb{N}=(N, +, \cdot,\dotdiv, \mathsf{s}, 0, 1)$ be the bag semiring expanded with the truncated subtraction $\dotdiv$ as monus and with support $\mathsf{s}$. Then there is 
no expression of  basic relational algebra $\mathcal{BRA}$ that defines 
the division operation
$\div$ on $\mathbb{N}$-databases  (i.e.,  division is not expressible using union, difference, Cartesian product, projection, selection,  and support on $\mathbb{N}$-databases).
\end{Thm}

\begin{proof} 
Let $\tau=(R,S)$ be a schema consisting of  a binary relation symbol $R$ and a unary relation symbol $S$.
For every $n\geq 1$, let $a, b_1,\ldots,b_n$ be distinct elements from the domain $D$ and 
let $\mathbf{I}(n)=(R_n,S_n)$ be the bag database ($\mathbb{N}$-database), where $R_n$ and $S_n$ are the following bags (${\mathbb{N}}$-relations):
\begin{itemize}
    \item $R_n(a,b_i) = 2$, for $1\leq i\leq n$, while $R_n(c,d)=0$, for all other pairs of elements of $D$.
    \item $S_n(b_i)=1$, for $1\leq i\leq n$, while $S_n(c)=0$, for every other element of $D$.
\end{itemize}
Clearly, for every $n\geq 1$, the division $R_n\div S_n$ is the bag such that
$(R_n\div S_n)(a) = 2^n$, while ($R_n\div S_n)(c) = 0$, for all other elements of $D$.

We will show that  no $\mathcal{BRA}$-expression 
defines $R\div S$   by showing that, for sufficiently large $n$, no such expression involving $R_n$ and $S_n$ can have an element of multiplicity $2^n$.

Let $E$ be an expression of the basic relational calculus $\mathcal{BRC}$. 
\begin{itemize}
\item The \emph{length} $l(E)$ of $E$  is defined by induction  as follows: if $E$ is $R$ or $S$, then $l(E)=1$; if  $E$ is $E_1 ~\mathrm{op} ~ E_2$, where $\mathrm{op}$ is union, difference, or Cartesian product, then $l(E)=l(E_1)+l(E_2)+1$; if $E$ is $\mathrm{op}(F)$, where $\mathrm{op}$ is projection,  selection, or support, then $l(E)=l(F)+1$.
\item We write $|\supp(E^{\mathbf{I}(n)})|$ to denote the number of tuples in the support %$\supp(E^{\mathbf{I}(n)})$ 
of the bag 
$E^{\mathbf{I}(n)}$
obtained by evaluating $E$ on the bag database $\mathbf{I}(n)$.
\item We write $\hm(E^{\mathbf{I}(n)})$ to denote the highest multiplicity of a tuple in the bag  $E^{\mathbf{I}(n)}$.
\end{itemize}
\smallskip
\noindent{\bf Claim:} The following statements are true:
\begin{enumerate}
\item $ |\supp(E^{\mathbf{I}(n)})| \leq n^{l(E)}$, for all $n\geq 2$.
\item There is a polynomial $p_E(n)$ such that
$\hm(E^{\mathbf{I}(n)})\leq p_E(n)2^{l(E)}$, for all  $n\geq 2$.
\end{enumerate}
Both statements  are proved by induction on the construction of the $\mathcal{BRC}$-expression $E$; the proof of the second statement uses also the first statement.

For the first statement, if $E$ is $R$ or $E$ is $S$, then $|\supp(E^{\mathbf{I}(n)})| = n = n^{l(E)}$, since 
in this case $l(E)=1.$ For the induction steps, the more interesting cases those of union and Cartesian product.  If $E=E_1 \cup E_2$, then
$|\supp(E^{\mathbf{I}(n)})|\leq |\supp(E_1^{\mathbf{I}(n)})|+ |\supp(E_2^{\mathbf{I}(n)})| \leq n^{l(E_1)}+ n^{l(E_2)} \leq 2n^{l(E_1)+l(E_2)}\leq n^{l(E_1)+l(E_2)+1}=n^{l(E)},$ where the second inequality follows from the induction hypothesis, while the one before the last one follows from the assumption that $n\geq 2$.
If $E=E_1 \times E_2$, then
$|\supp(E^{\mathbf{I}(n)})|= |\supp(E_1^{\mathbf{I}(n)})| \cdot |\supp(E_2^{\mathbf{I}(n)})| \leq n^{l(E_1)} \cdot n^{l(E_2)} =  n^{l(E_1)+l(E_2)}\leq n^{l(E_1)+l(E_2)+1}=n^{l(E)},$ where the first inequality follows from the induction hypothesis.

For the second statement,  if $E$ is $R$ or $E$ is $S$, then $\hm(E^{\mathbf{I}(n)}) \leq 2$, so we can put $p_E(n)=1$.  For the induction steps, the more interesting cases are those of  Cartesian product and  projection.
If $E=E_1 \times E_2$, then $\hm(E^{\mathbf{I}(n)}) = \hm(E_1^{\mathbf{I}(n)}) \hm(E_2^{\mathbf{I}(n)})\leq p_{E_1}(n)2^{l(E_1)}  p_{E_2}(n)2^{l(E_2)}= 
 p_{E_1}(n) p_{E_2}(n) 2^{l(E_1)+l(E_2)}\leq 
 p_{E_1}(n) p_{E_2}(n) 2^{l(E_1)+l(E_2)+1}=
  p_{E_1}(n) p_{E_2}(n)2^{l(E)}$, where the first inequality follows from the induction hypothesis.  So, in this case, we can  put $p_E(n) =  p_{E_1}(n) p_{E_2}(n)$. If $E$ is $\pi_v(F)$, where $\pi$ is the projection operation and $v$ is a (possibly empty) sequence of indices, then 
  $\hm(E^{\mathbf{I}(n)})=\hm((\pi_v(F))^{\mathbf{I}(n)})\leq |\supp(F^{\mathbf{I}(n)})| \hm(F^{\mathbf{I}(n)})\leq n^{l(F)}p_F(n)2^{l(F)}\leq n^{l(F)}p_F(n)2^{l(F)+1} = n^{l(F)}p_F(n)2^{l(E)}$, where the first inequality follows from the first part of the claim and the induction hypothesis. Thus, in this case, we can put $p_E( n) = n^{l(F)}p_F(n)$. %This completes the proof of the claim.
  
  \smallskip
  The second part of the preceding claim implies easily the theorem we are after. Indeed, if the division $R\div S$ could be expressed by some $\mathcal{BRA}$-expression $E$, then, by the second statement in the above claim,  there is a polynomial $p_E(n)$ such that for every $n\geq 2$, it is the case that $\hm(E^{\mathbf{I}(n)})\leq p_E(n)2^{l(E)}$. However, we also have that
  $\hm(E^{\mathbf{I}(n)})=2^n$, which, for sufficiently large $n$, is greater than $p_E(n)2^{l(E)}$, since $p_E(n)$ is some fixed polynomial and $2^{l(E)}$ is a constant.
\end{proof}

\begin{Rmk}
The support operation cannot be expressed in terms of the five basic relational algebra operations $\pi, \sigma, \cup, \times, \setminus$ and $\div$.
To see this, consider a relation symbol $R$ of arity 1 and the bag that has only the element a in its support and with multiplicity 2.  
 Let $E$ be an expression built from $R$ using the five basic relational algebra operations and division.   Either $E$ denotes the empty bag or $E$ denotes a bag of arity $r$ for some $r \geq 1$ whose support contains only a tuple of the form
$(a,a,\dots, a)$ and with an even number as multiplicity. The proof is by  induction on the construction of $E$.  
\end{Rmk}

\subsection{Relational calculus with universal quantifier}

For the remainder of this section, we assume that $\mathbb{K}$ is a positive semiring with monus and support. We now extend the basic relational calculus $\mathcal{BRC}$ with a universal quantifier $\forall$ and call the resulting system \emph{relational calculus} 
 $\mathcal{RC}$. Specifically,
we add the following clause to the syntax in Definition \ref{brc-synt-def}:

\begin{itemize}
    \item If $\varphi$ is an $\mathcal{RC}(\tau)$-formula with $y_1,\dots,y_n$ as its free variables, then, for every $i$ with $1\leq i\leq n$, the expression
$\forall y_i\varphi$ is a  $\mathcal{RC}(\tau)$-formula with the variables $y_1,\ldots,y_n$ other than $y_i$ as its free variables.
\end{itemize}
Furthermore, we add the following clause to the semantics in Definition~\ref{FOL-sem-def}:
\begin{itemize}
%\item[] $\|x\|_{\mathbf{M},v}=v(x)$;
%
%\item[] $\|F(t_1,\ldots,t_n)\|_{\mathbf{M},v}=F_{\mathbf{M}}(\|t_1\|_{\mathbf{M},v},\ldots,\|t_n\|^{{\alg{A}}}_{\mathbf{M},v})$,  for each $F\in Func_{\tau}$;

% \item[] $\|(\nabla \varphi )[{\bf a}]\|_{\model{A}}=\begin{cases} 1 & \text{if } \| \varphi_1 [{\bf a}]\|_{\model{A}} \neq 0 \\ 0 & \text{otherwise} \end{cases} $

 \item $\|\forall y_i\varphi\| (\mathbf{A}, \alpha) = \prod_{b\in A}\| \varphi\|(\mathbf{A}, \alpha[y_i/b])$.

\end{itemize}

As in the case of the existential quantifier, the semantics of the universal quantifier is well defined because of our blanket assumption that the structure $\mathbf{A}$ has a finite universe.

\paragraph*{Connections to non-classical logics}  \label{connections:sec}

It should be noted that the syntax and the semiring semantics of relational calculus $\mathcal{RC}$ that we defined here  have also been considered in the context of algebraic semantics of non-classical logics. Specifically, there are  connections to dual intuitionistic logic that we now discuss.

First-order intuitionistic logic has $\wedge, \vee, \rightarrow, 0$ as primitive connectives and $\forall$ and $\exists$ as quantifiers.  The connective $\rightarrow$ is \emph{intuitionistic implication}; it is used to define the negation of a formula as $\neg \varphi := \varphi \rightarrow 0$. Algebraic semantics is one of several different approaches to defining rigorous semantics of intuitionistic logic. The main idea, which originated with Heyting
\cite{heyting1930formalen}, is to consider
the class $\mathcal{H}$ of all 
\emph{Heyting algebras}, where a Heyting algebra is a complete bounded distributive lattices expanded with a binary operation $\Rightarrow$ that has the following property: $(a\wedge b) \preceq c$ if and only if $a\preceq (b \Rightarrow c)$. If $\mathbb{H}$ is a Heyting algebra and $\varphi$ is a formula of first-order intuitionistic logic, then the semantics $\|\varphi\|$ of $\varphi$ on $\mathbb{H}$ is defined using $\Rightarrow$ to interpret the intuitionistic implication $\rightarrow$. A sentence $\psi$ of intuitionistic logic is \emph{valid} if for every Heyting algebra $\mathbb{K}$ and every (finite or infinite) $\mathbb{K}$-structure $\mathbf{A}$, we have that $\psi^\mathbf{A}=1$. %Thus, Heyting algebras are to intuitionistic logic what Boolean algebras are to classical logic.
Soundness and completeness theorems are then established about proof systems that are capable to derive precisely the valid sentences of intuitionistic logic.

Dual first-order intuitionistic logic has the same syntax as relational calculus $\mathcal{RC}$, except that the support $\nabla$ is not included as a connective. Informally, dual intuitionistic logic is obtained from intutitionistic logic by replacing the $\Rightarrow$ connective  by the $\butnot$ connective. Goodman~\cite{Good} gave algebraic semantics to dual first-order intuitionistic logic by considering the class $\mathcal{B}$ of all Brouwerian algebras discussed in Example~\ref{latt}. The semantics $\|\varphi\|$ of a formula $\varphi$ of dual first-order intuitionistic logic on a Brouwerian algebra $\mathbb{K}$ given by Goodman is precisely the semantics of $\varphi$ as an $\mathcal{RC}$-formula (without $\nabla$) on $\mathbb{K}$ that we gave here. In particular, the  connective $\butnot$ is interpreted using the monus of $\mathbb{K}$.

\subsection{Codd's Theorem for Relational Algebra and Relational Calculus}

We now establish a version of Codd's Theorem for relational algebra  $\mathcal{RA}$
and  relational calculus $\mathcal{RC}$.

\begin{Thm}\label{Codd2} 
Let\/ $\mathbb{K}$ be a positive semiring with monus and support, let $\tau=(\hat{R_1},\ldots \hat{R_p})$ be a schema, and let $q$ be an $n$-ary query. The following
statements are equivalent:
\begin{itemize}
\item[1.] There is an $\mathcal{RA}(\tau)$-expression $E$ of arity $n$  such that
$q^\mathbf{I} = E^\mathbf{I}$, for every $\mathbb{K}$-database  $\mathbf{I}$.
\item[2.] There is a domain independent $\mathcal{RC}(\tau)$-formula $\varphi(y_1,\ldots,y_n)$ such
that $q^\mathbf{I}= \varphi^{\mathbf{A}(\mathbf{I})}$, for every $\mathbb{K}$-database  $\mathbf{I}$.
\item[3.] There is an $\mathcal{RC}(\tau)$-formula
$\varphi(y_1,\ldots,y_n)$ such that $q^\mathbf{I}= \varphi^{\mathbf{A}(\mathbf{I})}$, for every $\mathbb{K}$-database  $\mathbf{I}$.
\end{itemize}
\end{Thm}

\begin{proof}

The proof expands that of Theorem~\ref{Codd1}. The direction $(2) \implies (3)$ is obvious, so we only prove the other two directions.   We  focus on the case of  the division operation in direction $(1) \implies (2)$, and on the case of universal quantification in direction $(3) \implies (1)$. All other cases are identical to the ones  in the proof of Theorem~\ref{Codd1}.

$(1) \implies (2):$  
Assume that $E$ is an $\mathcal{RA}(\tau)$-expression of the form $E_1 \div E_2$, where $E_1$ and $E_2$ are $\mathcal{RA}(\tau)$-expressions of arities  $n_1$ and $n_2$ with  $n_1 > n_2$. By induction hypothesis and by renaming  variables if necessary, there are  domain independent $\mathcal{RC}(\tau)$-formulas $\varphi_{1}(y_1,\dots, y_{n_1})$ and $ \varphi_2(y_{n_1-n_2+1},\ldots, y_{n_1})$, 
such that for every $\mathbb K$-database $\mathbf{I}$, we have that
$E_1^{\mathbf I} = \varphi_1^{\mathbf{A}(\mathbf{I})}$ and
$E_2^{\mathbf I} = \varphi_2^{\mathbf{A}(\mathbf{I})}$. We will  construct a domain independent
$\mathcal{RC}(\tau)$-formula $\varphi_E$ such that for every $\mathbb K$-database $\mathbf I$, we have that
$E^{\mathbf I} = \varphi_E^{\mathbf{A}(\mathbf{I})}$. 
%In addition to containing $\varphi_1$ and $\varphi_2$ as building blocks, $\varphi_E$ will have two other building blocks that we now introduce.

Let $\eta$ be the $\mathcal{RC}(\tau)$-sentence  $\nabla \exists y_1(y_1=y_1)$. It is easy to see that $\eta$ has the following property:
\begin{itemize}
\item 
For every $\mathbb K$-database $\mathbf{I}=(R_1, \dots, R_p)$  and for every 
$\mathbb K$-structure
  $\mathbf{B}=(B,R_1, \dots, R_p)$  with the same relations as $\mathbf I$ and with 
  $\mathrm{adom}(\mathbf{I})\subseteq B\subseteq D$, we have that 
  $\eta^{\mathbf B}(<>) = 1$.
%\item $\eta$ is domain independent.
  \end{itemize}
Indeed, this property holds because
$$\eta^{\mathbf B}(<>)= \|\nabla \exists y_1(y_1=y_1)\|({\mathbf B},<>)= \mathsf{s}(\sum_{b\in B}(\|y_1=y_1\|({\mathbf B},b)))=1, $$
where the first two equalities follow from Definitions \ref{FOL-sem-def}
and \ref{brc-eval-def}, while the third equality holds for the following reason: from the blanket assumption that  $\mathbf I$ is non-trivial, we have that $\mathrm{adom}({\mathbf I})\not = \emptyset$, hence $B\not = \emptyset$,  and so  
$\sum_{b\in B}(\|y_1=y_1\|({\mathbf B},b))= \sum_{b\in B}1 \not = 0$, since $\mathbb K$ is a zero-sum free semiring.
%Note  that this argument also shows that
%$\eta^{\mathbf{A}(\mathbf{I})}(<>)=1$, hence the sentence $\eta$ is domain independent\footnote{Clearly, the sentence $\eta$ involves the equality symbol $=$. 
\footnote{
Clearly, the $\mathcal {RC}$-sentence $\eta$ uses equality $=$. There are equality-free $\mathcal{RC}$-sentences that have the properties of $\eta$. For example, it is easy to verify that $\nabla(\bigvee_{i=1}^p \exists y_1\ldots
\exists y_{r_i}\hat{R_i}(y_1,\ldots,y_{r_i}))$ is such a sentence.}

\commentout{
In addition to having the $\mathcal{RC}(\tau)$-formulas $\varphi_1$ and
$\varphi_2$ as building blocks, the formula
$\varphi_E$ will also have another $\mathcal{RC}(\tau)$-formula  as a building block. We describe that other formula next.

We claim that there is a domain independent $\mathcal{RC}(\tau)$-sentence $\theta_\tau$ (i.e., $\theta$ has no free variables) such that if $\mathbf{I}=(R_1, \dots, R_p)$ is a $\mathbb K$-database and 
  $\mathbf{B}=(B,R_1, \dots, R_p)$  is a $\mathbb{K}$-structure with $\mathrm{adom}(\mathbf{I})\subseteq B\subseteq D$, then
\begin{center}
$\theta_\tau^{\mathbf{B}}(< >) =
\begin{cases}  1 & \mbox{if $\mathrm{adom}(\mathbf{I}) \neq \emptyset$}\\
0 & \mbox{otherwise.}
\end{cases}
$
\end{center}

 Let
$\theta_\tau$ be the $\mathcal{RC}(\tau)$-sentence 
$\nabla(\bigvee_{i=1}^p \exists y_1\ldots
\exists y_{r_i}\hat{R_i}(y_1,\ldots,y_{r_i}))$.
By Definitions \ref{FOL-sem-def} and \ref{brc-eval-def}, we have 
 $$\theta_\tau^\mathbf{B}(<>)=  \|\theta_\tau\|(\mathbf{B},< >) = 
\|\nabla(\bigvee_{i=1}^p \exists y_1\ldots
\exists y_{r_i}\hat{R_i}(y_1,\ldots,y_{r_i}))\|(\mathbf{B},< >)=
$$
\begin{center}
$
\mathsf{s} (\|\bigvee_{i=1}^p \exists y_1\ldots
\exists y_{r_i}\hat{R_i}(y_1,\ldots,y_{r_i})\|(\mathbf{B},< >)) = 
\begin{cases} 1& \mbox{if $\|\bigvee_{i=1}^p \exists y_1\ldots
\exists y_{r_i}\hat{R_i}(y_1,\ldots,y_{r_i})\|(\mathbf{B},< >) \neq 0$} \\
0 & \mbox{otherwise.} \end{cases}$
\end{center}
Moreover, $\|\bigvee_{i=1}^p \exists y_1\ldots
\exists y_{r_i}\hat{R_i}(y_1,\ldots,y_{r_i})\|(\mathbf{B},< >) = \sum_{i=1}^p \|\exists y_1\ldots
\exists y_{r_i}\hat{R_i}(y_1,\ldots,y_{r_i})\|(\mathbb{B},<>)$. Since the relations of $\mathbb B$ are the same as the relations of $\mathbb I$, we have that
$$\sum_{i=1}^p \|\exists y_1\ldots
\exists y_{r_i}\hat{R_i}(y_1,\ldots,y_{r_i})\|(\mathbb{B},<>)=
\sum_{i=1}^p \sum_{b_1,\ldots,b_{r_i} \in B} R_i(b_1,\ldots,b_{r_i}).$$
%\sum_{i=1}^p \sum_{b_1,\ldots,b_{r_i} \in \mathrm{adom}(\mathbf I)} R_i(b_1,\ldots,b_{r_i}).$$
Since $\mathbb K$ is a zero-sum free semiring, we have that $\sum_{i=1}^p \|\exists y_1\ldots
\exists y_{r_i}\hat{R_i}(y_1,\ldots,y_{r_i})\|(\mathbb{B},<>) \not = 0$ if and only if 
$R_i(b_1,\ldots,b_{r_i})\not = 0$, for some $i\leq p$ and some elements $b_1,\ldots,b_{r_i}$ in $B$. Since  $\mathrm{adom}(\mathbf{I})\subseteq B\subseteq D$, it follows that 
$\sum_{i=1}^p \|\exists y_1\ldots
\exists y_{r_i}\hat{R_i}(y_1,\ldots,y_{r_i})\|(\mathbb{B},<>) \not = 0$ if and only if 
$\mathrm{adom}(\mathbf{I}) \not = \emptyset$. Consequently, $\theta_\tau$ has the claimed properties, that is, 
$\theta_\tau^{\mathbf B}(<>) = 1$ if $\mathrm{adom}(\mathbf{I}) \not = \emptyset$, and $\theta_\tau^{\mathbf B}(<>) = 0$ if
$\mathrm{adom}(\mathbf{I}) = \emptyset$. Furthermore, since these  properties hold for every structure ${\mathbf B}=(B,R_1,\ldots,R_p)$ with
$\mathrm{adom}(\mathbf{I})\subseteq B\subseteq D$ and since
${\mathbf A}({\mathbf I})$ is such a structure, it  follows that $\theta_\tau^{\mathbf B}= \theta_\tau^{{\mathbf A}({\mathbf I})}$, hence 
$\theta_\tau$ is domain independent.
%endcommentout
}

With the $\mathcal{RC}(\tau)$-sentence $\eta$   at hand, we take  $\varphi_E$  to be the $\mathcal{RC}(\tau)$-formula $$   
(\nabla \exists y_{n_1-n_2+1}\ldots \exists y_{n_1}\varphi_1) \wedge (\forall y_{n_1-n_2+1} \ldots \forall  y_{n_1} ((  \eta \butnot \nabla \varphi_{2}) \vee   (  \nabla\varphi_{2} \wedge \varphi_{1} ))),$$
 whose free variables are
 $y_1,\ldots,y_{n_1-n_2}$. We will show that $\varphi_E$ is domain independent and  that $E^{\mathbf I} = \varphi_E^{\mathbf{A}(\mathbf{I})}$
 holds
for every $\mathbb K$-database ${\mathbf I}  $.  We begin with the latter task.

 We must show that for all  $a_1,\ldots,a_{n_1-n_2} \in D$, it holds that
 $E^{\mathbf I}( a_1,\ldots,a_{n_1-n_2}) = \varphi_E^{\mathbf{A}(\mathbf{I})}(a_1,\ldots,a_{n_1-n_2})$.
Assume first that $a_1,\ldots,a_{n_1-n_2} \in \mathrm{adom}({\mathbf I})$.
By putting
${\bf a}=(a_1,\ldots,a_{n_1-n_2})$ and ${\bf y}= (y_{n_1-n_2+1} \ldots   y_{n_1})$, we have that
$$\varphi_E^{\mathbf{A}(\mathbf{I})}({\bf a}) = \|
(\nabla \exists {\bf y} \varphi_1) \wedge 
(\forall {\bf y} ((  \eta \butnot \nabla \varphi_{2}) \vee   (  \nabla\varphi_{2} \wedge \varphi_{1} )))\|(\mathbf{A}(\mathbf{I}), {\bf a})= $$

$$\|\nabla \exists {\bf y} \varphi_1\|(\mathbf{A}(\mathbf{I}), {\bf a}) ~ \cdot 
\prod_{{\bf b} \in \mathrm{adom}({\mathbf I})^{n_2}}\| (\eta \butnot \nabla \varphi_{2}) \vee   (  \nabla\varphi_{2} \wedge \varphi_{1} )\|(\mathbf{A}(\mathbf{I}), {\bf a},{\bf b}),
$$
where the first equality follows from the definition of the formula $\varphi_E$ and 
the second equality
follows from Definitions \ref{FOL-sem-def}
and \ref{brc-eval-def} (and  their extensions for formulas with universal quantification). We now distinguish two cases, namely, whether or not there is a tuple ${\bf b}$ such
that $\varphi_1^{\mathbf{A}({\mathbf I})}({\bf a},{\bf b})\not = 0$.

If for all tuples $\bf b$, we have that $\varphi_1^{\mathbf{A}({\mathbf I})}({\bf a},{\bf b}) = 0$, then
$\|\nabla \exists {\bf y} \varphi_1\|(\mathbf{A}(\mathbf{I}), {\bf a})= 0$, hence $\varphi_E^{\mathbf{A}(\mathbf{I})}({\bf a}) =0$.
By induction hypothesis about $E_1$ and $\varphi_1$, 
we have that $E_1^{\mathbf I}({\bf a},{\bf b})=0$, for 
all tuples $\bf b$. Consequently,  $\mathsf{s}(\sum_{{\bf b} \in D^{n_2}}E_1^{\mathbf I}({\bf a},{\bf b}))=0$, hence, by Definition \ref{def:div}, 
$E_1^{\mathbf I}\div E_2^{\mathbf I}({\bf a})=0$ and so $E^{\mathbf I}({\bf a})=
0$.
 Thus, in this case
$\varphi_E^{\mathbf{A}(\mathbf{I})}({\bf a}) = E^{\mathbf I}({\bf a})$.

If there is a tuple $\bf b$ such that $\varphi_1^{\mathbf{A}({\mathbf I})}({\bf a},{\bf b})\not = 0$, then
$\|\nabla \exists {\bf y} \varphi_1\|(\mathbf{A}(\mathbf{I}), {\bf a})=1$.
Therefore, we have that
$$\varphi_E^{\mathbf{A}(\mathbf{I})}({\bf a}) =\prod_{{\bf b}\in \mathrm{adom}({\mathbf I})^{n_2}}(\| \eta \butnot \nabla \varphi_{2})\|(\mathbf{A}(\mathbf{I}),{\bf b}) + \|(\nabla\varphi_{2} \wedge \varphi_{1} ))\|(\mathbf{A}(\mathbf{I}), {\bf a},{\bf b}))=$$
$$
\prod_{{\bf b}\in \mathrm{adom}({\mathbf I})^{n_2}} (\eta^{\mathbf{A}(\mathbf{I})}(<>) -(\nabla \varphi_2)^{\mathbf{A}(\mathbf{I})}({\bf b}) + 
(\nabla \varphi_2)^{\mathbf{A}(\mathbf{I})}({\bf b})\cdot  \varphi_1^{\mathbf{A}(\mathbf{I})}({\bf a},{\bf b}))= 
$$
$$\prod_{{\bf b}\in \mathrm{adom}({\mathbf I})^{n_2}}(1-(\nabla \varphi_2)^{\mathbf{A}(\mathbf{I})}({\bf b}) + 
(\nabla \varphi_2)^{\mathbf{A}(\mathbf{I})}({\bf b})\cdot  \varphi_1^{\mathbf{A}(\mathbf{I})}({\bf a},{\bf b}))
,$$
where the second equality
follows from Definitions \ref{FOL-sem-def}
and \ref{brc-eval-def}, and the third equality follows  from the properties of the sentence $\eta$. Fix a  tuple
${\mathbf b} \in \mathrm{adom}({\mathbf I})^{n_2}$ and consider the expression
$$1-(\nabla \varphi_2)^{\mathbf{A}(\mathbf{I})}({\bf b}) + 
(\nabla \varphi_2)^{\mathbf{A}(\mathbf{I})}({\bf b})\cdot \varphi_1^{\mathbf{A}(\mathbf{I})}({\bf a},{\bf b}).$$
If $\varphi_2^{\mathbf{A}(\mathbf{I})}({\bf b}) = 0$, then $(\nabla \varphi_2)^{\mathbf{A}(\mathbf{I})}({\bf b}) = 0$, hence the above expression evaluates to 1, since $1-0=1$
and $0\cdot x=0$, for every element $x$ in the universe of the semiring $\mathbb K$.
If $\varphi_2^{\mathbf{A}(\mathbf{I})}({\bf b}) \not = 0$, then $(\nabla \varphi_2)^{\mathbf{A}(\mathbf{I})}({\bf b}) = 1$, hence the above expression evaluates to $\varphi_1^{\mathbf{A}(\mathbf{I})}({\bf a},{\bf b})$,  since $1-1=0$.
From the preceding analysis, it follows that
$$\varphi_E^{\mathbf{A}(\mathbf{I})}({\bf a}) = \prod_{{\bf b}~:~\varphi_2^{\mathbf{A}(\mathbf{I})}({\bf b}) \not = 0} \varphi_1^{\mathbf{A}(\mathbf{I})}({\bf a},{\bf b}). $$
By induction hypothesis, we have that $\varphi_1^{\mathbf{A}(\mathbf{I})}({\bf a},{\bf b})=E_1^{\mathbf I}({\bf a},{\bf b})$ and $\varphi_2^{\mathbf{A}(\mathbf{I})}({\bf b})= E_2^{\mathbf I}({\bf b})$. 
Furthermore, $\mathsf{s}(\sum_{{\bf b} \in D^{n_2}}E_1^{\mathbf I}({\bf a},{\bf b}))=1$, because
$\|\nabla \exists {\bf y} \varphi_1\|(\mathbf{A}(\mathbf{I}), {\bf a})=1$.
Therefore, we have that
$$\varphi_E^{\mathbf{A}(\mathbf{I})}({\bf a}) = \mathsf{s}(\sum_{{\bf b} \in D^{n_2}}E_1^{\mathbf I}({\bf a},{\bf b})) \cdot \prod_{{\bf b}~:~E_2^{\mathbf I}({\bf b}) \not = 0} E_1^{\mathbf I}({\bf a},{\bf b}) = (E_1^{\mathbf I}\div E_2^{\mathbf I})({\bf a})= E^{\mathbf I}({\bf a}).$$

Suppose now that  ${\bf a}=(a_1, \dots, a_{n_1-n_2})\not \in \mathrm{adom}({\mathbf I})^{n_1-n_2}$.
In this case, we have
that $\varphi_E^{\mathbf{A}(\mathbf{I})}=0$. Furthermore, for every tuple  ${\bf b}$, we have
that $({\bf a},{\bf b})\not \in \mathrm{adom}({\mathbf I})^{n_1}$, hence
$\varphi_1^{\mathbf{A}(\mathbf{I})}({\bf a},{\bf b})=0$. By induction hypothesis about $\varphi_1$ and $E_1$, we have that
$E_1^{\mathbf I}({\bf a},{\bf b})=0$  for every tuple
$\bf b$, hence $\mathsf{s}(\sum_{{\bf b} \in D^{n_2}}E_1^{\mathbf I}({\bf a},{\bf b}))=0$ and so. by Definition \ref{def:div},
$E^{\mathbf I}({\bf a})=(E_1^{\mathbf I}\div E_2^{\mathbf I})({\bf a}) = 0$. Thus, once again, we have that $\varphi_E^{\mathbf{A}(\mathbf{I})}({\bf a}) = E^{\mathbf I}({\bf a})$.
This completes the proof of the correctness of the formula
$\varphi_E$.

Next, we show that $\varphi_E$ is domain independent. For this, we have to show that
for every $\mathbb K$-database $\mathbf{I}=(R_1, \dots, R_p)$  and for every 
$\mathbb K$-structure
  $\mathbf{B}=(B,R_1, \dots, R_p)$  with the same relations as $\mathbf I$ and with 
  $\mathrm{adom}(\mathbf{I})\subseteq B\subseteq D$, we have
  that $\varphi_E^{\mathbf B}= \varphi_E^{\mathbf{A}({\mathbf I})}$, which means that  $\varphi_E^{\mathbf B}({\bf a})= \varphi_E^{\mathbf{A}({\mathbf I})}({\bf a})$ holds, for every tuple ${\bf a} \in D^{n_1-n_2}$. 
  We distinguish three cases for a tuple ${\bf a}$ in $D^{n_1-n_2}$. 

\smallskip

 \noindent{\bf{Case 1:}} 
${\bf a} \not \in B^{n_1-n_2}$. Since $\mathrm{adom}({\mathbf I})\subseteq B$, we have that ${\bf a} \not \in \mathrm{adom}({\mathbf I})^{n_1-n_2}$, as well. Hence, by Definition \ref{brc-eval-def}, we have that 
$\varphi_E^{\mathbf B}({\bf a})=0 =  \varphi_E^{\mathbf{A}({\mathbf I})}({\bf a})$.

\smallskip

 \noindent{\bf{Case 2:}} 
${\bf a} \in B^{n_1-n_2}$ and for every tuple ${\bf b} \in D^{n_2}$, we have that
$\varphi_1^{\mathbf B}({\bf a},{\bf b})=0$.
Since ${\bf a}\in B^{n_1-n_2}$,  we have that
$$\varphi_E^{\mathbf{B}}({\bf a}) = \|
(\nabla \exists {\bf y} \varphi_1) \wedge 
(\forall {\bf y} ((  \eta \butnot \nabla \varphi_{2}) \vee   (  \nabla\varphi_{2} \wedge \varphi_{1} )))\|(\mathbf{B}, {\bf a})= $$
$$\|\nabla \exists {\bf y} \varphi_1\|(\mathbf{B}, {\bf a}) ~ \cdot 
\prod_{{\bf b}\in B^{n_2}}(\| (\eta \butnot \nabla \varphi_{2}) \vee   (  \nabla\varphi_{2} \wedge \varphi_{1} ))\|(\mathbf{B}, {\bf a},{\bf b}).
$$
Since $\varphi_1^{\mathbf B}({\bf a},{\bf b})=0$ for every tuple ${\bf b}\in B^{n_2}$, it follows that
$\|\nabla \exists {\bf y} \varphi_1\|(\mathbf{B}, {\bf a}) =0$, hence $\varphi_E^{{\mathbf B}}({\bf a})=0$. Let us now consider $\varphi_E^{{\mathbf A}({\mathbf I})}({\bf a})$.
If ${\bf a}\not \in \mathrm{adom}({\mathbf I})^{n_1-n_2}$, then, by Definition \ref{brc-eval-def}, we have that 
$\varphi_E^{{\mathbf A}({\mathbf I})}({\bf a})=0$.   If  ${\bf a} \in \mathrm{adom}({\mathbf I})^{n_1-n_2}$, then by the same calculations as above we have that
$$\varphi_E^{{\mathbf A}({\mathbf I})}({\bf a}) = 
\|\nabla \exists {\bf y} \varphi_1\|(\mathbf{A}(\mathbf{I}), {\bf a}) ~ \cdot 
\prod_{{\bf b}\in \mathrm{adom}({\mathbf I})^{n_2}}(\| (\eta \butnot \nabla \varphi_{2}) \vee   (  \nabla\varphi_{2} \wedge \varphi_{1} ))\|(\mathbf{A}(\mathbf{I}), {\bf a},{\bf b}).
$$
Since $\varphi_1$ is domain independent and since $\varphi_1^{\mathbf B}({\bf a},{\bf b})=0$ for every tuple ${\bf b}\in B^{n_2}$, we have that
$\varphi_1^{{\mathbf A}({\mathbf I})}({\bf a},{\bf b})=0$ for every tuple ${\bf b}\in \mathrm{adom}({\mathbf I})^{n_2}$, hence 
$\|\nabla \exists {\bf y} \varphi_1\|(\mathbf{A}(\mathbf{I}), {\bf a})= 0$ and so $\varphi_E^{{\mathbf A}({\mathbf I})}({\bf a}) =  0 =
\varphi_E^{\mathbf B}({\bf a})$.

\smallskip

 \noindent{\bf{Case 3:}} 
${\bf a} \in B^{n_1-n_2}$ and there is a  tuple ${\bf b} \in D^{n_2}$ such that
$\varphi_1^{\mathbf B}({\bf a},{\bf b}) \not =0$. By the domain independence of $\varphi_1$, we have
that $\varphi_1^{{\mathbf A}({\mathbf I})}({\bf a},{\bf b}) \not =0$, hence the tuple
$({\bf a}, {\bf b})$ must belong to  $\mathrm{adom}({\mathbf I})^{n_1}$ and so the tuple ${\bf a}$ must belong to  $\mathrm{adom}({\mathbf I})^{n_1-n_2}$. 
Furthermore, we have that 
$\|\nabla \exists {\bf y} \varphi_1\|(\mathbf{B}, {\bf a})= 1 = 
\|\nabla \exists {\bf y} \varphi_1\|(\mathbf{A}(\mathbf{I}), {\bf a})= 1$,
 by the domain independence of $\nabla \varphi_1$ (which follows from the domain independence of $\varphi_1$ and the third part of Proposition \ref{pro:dom-ind}). Using the same calculations as above,
 we have that
 $$\varphi_E^{{\mathbf B}}({\bf a}) =
\prod_{{\bf b}\in {\mathbf B}^{n_2}}\| (\eta \butnot \nabla \varphi_{2}) \vee   (  \nabla\varphi_{2} \wedge \varphi_{1} ))\|(\mathbf{B}, {\bf a},{\bf b})
$$
and
$$\varphi_E^{{\mathbf A}({\mathbf I})}({\bf a}) = 
\prod_{{\bf b}\in \mathrm{adom}({\mathbf I})^{n_2}}\| (\eta \butnot \nabla \varphi_{2}) \vee   (  \nabla\varphi_{2} \wedge \varphi_{1} ))\|(\mathbf{A}(\mathbf{I}), {\bf a},{\bf b}).
$$
Thus, to show that 
$\varphi_E^{{\mathbf B}}({\bf a}) = \varphi_E^{{\mathbf A}({\mathbf I})}({\bf a})$, it remains to show that the expressions in the right-hand sides of the previous two equations have the same value.
Fix a tuple ${\bf b}\in B^{n_2}$. If $\|\nabla \varphi_2\|({\mathbf B}, {\bf b})=0$, then $$\| (\eta \butnot \nabla \varphi_{2}) \vee   (  \nabla\varphi_{2} \wedge \varphi_{1} ))\|(\mathbf{B}, {\bf a},{\bf b}) = 1,$$ hence such tuples ${\bf b}$ do not contribute to the product and can be ignored. Suppose that $\|\nabla \varphi_2\|({\mathbf B}, {\bf b})=1$, which also means that $(\nabla \varphi_2)^{\mathbf B}({\bf b})=1$. Since $\varphi_2$ is domain independent, the third part of Proposition \ref{pro:dom-ind} implies that $\nabla \varphi_2$ is domain independent, hence 
$(\nabla \varphi_2)^{{\mathbf A}({\mathbf I})}({\bf b})=1$, which, in turn, implies that
${\bf b}$ belongs to $\mathrm{adom}({\mathbf I})^{n_2}$ (otherwise, $(\nabla \varphi_2)^{{\mathbf A}({\mathbf I})}({\bf b})=0$).
This analysis shows that the only tuples  that contribute to the right-hand sides of the two equations above are the tuples ${\bf b} \in \mathrm{adom}({\mathbf I})^{n_2}$  such that 
$ (\nabla \varphi_2)^{\mathbf B}({\bf b}) =  (\nabla \varphi_2)^{{\mathbf A}({\mathbf I})}({\bf b}) = 1$, which also means
that $ \|\nabla \varphi_2\|({\mathbf B},{\bf b}) =  \|\nabla \varphi_2\|({\mathbf A}({\mathbf I}),{\bf b}) = 1$.
For such tuples ${\bf b}$ and since ${\bf a} \in \mathrm{adom}({\mathbf I})^{n_1-n_2}$, we have that
{\small
$$\| (\eta \butnot \nabla \varphi_{2}) \vee   (  \nabla\varphi_{2} \wedge \varphi_{1} ))\|(\mathbf{B}, {\bf a},{\bf b}) = \|\eta \butnot \nabla \varphi_{2}\|({\mathbf B},{\bf b}) + 
\|\nabla\varphi_{2} \wedge \varphi_{1} \|({\mathbf B},{\bf a}, {\bf b}) = \|\varphi_1\|({\mathbf B},{\bf a}, {\bf b})$$
}
and
{\small
$$\| (\eta \butnot \nabla \varphi_{2}) \vee   (  \nabla\varphi_{2} \wedge \varphi_{1} ))\|(\mathbf{A}({\mathbf I}), {\bf a},{\bf b}) = \|\eta \butnot \nabla \varphi_{2}\|(\mathbf{A}({\mathbf I}),{\bf b}) + 
\|\nabla\varphi_{2} \wedge \varphi_{1} \|(\mathbf{A}({\mathbf I}),{\bf a}, {\bf b}) = \|\varphi_1\|(\mathbf{A}({\mathbf I}),{\bf a}, {\bf b}).$$}
By the domain independence of $\varphi_1$, we have that 
$\|\varphi_1\|({\mathbf B},{\bf a}, {\bf b})  = \|\varphi_1\|(\mathbf{A}({\mathbf I}),{\bf a}, {\bf b})$, which, in view of the preceding calculations, implies that 
$\varphi_E^{{\mathbf B}}({\bf a}) = \varphi_E^{{\mathbf A}({\mathbf I})}({\bf a})$ holds in Case 3, as well. 
This completes the proof of
the direction $(1) \implies (2)$.
Note that the building blocks of the formula
$\varphi_E$ do not have the same free variables; therefore,   by Remark \ref{rmk:dom-ind} one cannot rely entirely on  Proposition \ref{pro:dom-ind} and argue in a modular way that
$\varphi_E$ is domain independent.
\medskip

$(3) \implies (1):$
Assume that $\varphi$ is an 
$\mathcal{RC}(\tau)$-formula of the form 
$\forall y_i\psi$, where $\psi$ is
an $\mathcal{RC}(\tau)$-formula with $y_1,\ldots,y_n$ as its free variables and
$i$ is an index such that $1\leq i\leq n$. For notational simplicity, we assume that $i=n$, so that $\varphi$ is of the form
$\forall y_n\psi$, where the free variables of $\psi$ are
$y_1,\ldots,y_n$. By induction hypothesis, there is an $\mathcal{RA}(\tau)$-expression $E_\psi$ of arity $n$ such that $E_\psi^{\mathbf I}= \psi^{\mathbf{A}(\mathbf{I})}$, for every $\mathbb K$-database $\mathbf I$. We take $E_\varphi$ to be the 
$\mathcal{RA}$-expression
 $$E_\psi \div E_{\mathrm{adom}},$$ 
 which has arity $n-1$. We must show that
 $E_\varphi^{\mathbf I} = \varphi^{\mathbf{A}(\mathbf{I})}$, for every $\mathbf K$-database $\mathbf I$, which means that $E_\varphi^{\mathbf I}({\bf a}) = \varphi^{\mathbf{A}(\mathbf{I})}({\bf a})$
 holds, for every tuple ${\bf a}=(a_1,\ldots,a_{n-1})\in D^{n-1}$.
 
For every tuple ${\bf a}\in D^{n-1}$ and by the semantics of the division operation, we have that
$$E_\varphi^{\mathbf I}({\bf a})= (E_\psi \div E_{\mathrm{adom}})^{\mathbf I}({\bf a})= 
(E^{\mathbf I}_\psi \div E_{\mathrm{adom}}^{\mathbf I})({\bf a}) = 
\mathsf{s}(\sum_{b\in D}E^{\mathbf I}_{\psi}({\bf a },b)) ~ \cdot \prod_{E_{\mathrm{adom}}^
{\mathbf I}(b) \not = 0} E_\psi^{\mathbf {I}}({\bf a},b).
$$
By Proposition \ref{pro:active}, 
$E_{\mathrm{adom}}^
{\mathbf I}= \mathrm{adom}(\mathbf{I})$. Hence, for every $a\in D$, we have that
$E_{\mathrm{adom}}^
{\mathbf I}(a) \not = 0$ if and only if $a \in \mathrm{adom}(\mathbf{I})$.
Consequently, for every tuple ${\mathbf a} \in D^{n-1}$, we have that 
$$E_\varphi^{\mathbf I}({\bf a})= \mathsf{s}(\sum_{b\in D}E^{\mathbf I}_{\psi}({\bf a },b)) ~ \cdot 
\prod_{b \in \mathrm{adom}(
{\mathbf I})} E_\psi^{\mathbf {I}}({\bf a},b).
$$
 We now distinguish two cases for a tuple ${\bf a}\in D^{n-1}$.

\smallskip

\noindent{{\bf Case 1:}} For every element $b\in  D$, we have that  $\psi^{\mathbf{A}({\mathbf I})}({\bf a}, b)= 0$.  By induction hypothesis about $\psi$
and $E_\psi$, we have that 
$E_\psi^{\mathbf I}({\bf a},b) = \psi^{\mathbf{A}({\mathbf I})}({\bf a}, b)=0$ holds, for every $b\in D$. It follows that $\sum_{b\in D}E^{\mathbf I}_{\psi}({\bf a },b)=0$, hence, by the last displayed equation, we have that $E_\varphi^{\mathbf I}({\bf a})=0$. We claim that also 
$\varphi^{\mathbf{A}(\mathbf{I})}({\bf a}) =0$
 holds in this case.  
  Indeed, if ${\bf a}\not \in \mathrm{adom}({\mathbf I})$, then $\varphi^{\mathbf{A}(\mathbf{I})}({\bf a}) =0$, by Definition \ref{brc-eval-def}. If  ${\bf a} \in \mathrm{adom}({\mathbf I})$, then, by Definitions \ref{FOL-sem-def} and \ref{brc-eval-def}, we have that
  $$\varphi^{\mathbf{A}(\mathbf{I})}({\bf a}) = \|\forall y_i\psi\|(\mathbf{A}(\mathbf{I}),{\bf a}) = \prod_{b \in \mathrm{adom}({\mathbf I})}\|\psi\|(\mathbf{A}(\mathbf{I}),{\bf a})
  = \prod_{b \in \mathrm{adom}({\mathbf I})} \psi^{\mathbf{A}(\mathbf{I})}({\bf a},b)=
  0,$$
 as desired. Note that the product 
 $\prod_{b \in \mathrm{adom}({\mathbf I})} \psi^{\mathbf{A}(\mathbf{I})}({\bf a},b)$
 is non-empty because of our blanket assumption that $\mathrm{adom}({\mathbf I})$ is non-empty. 

 \smallskip
 \noindent{{\bf Case 2:}} There is an  element $b\in  D$ such that  $\psi^{\mathbf{A}({\mathbf I})}({\bf a}, b) \not = 0$.  By induction hypothesis about $\psi$
and $E_\psi$, we have that 
$E_\psi^{\mathbf I}({\bf a},b) = \psi^{\mathbf{A}({\mathbf I})}({\bf a}, b)$, hence $ E_\psi^{\mathbf I}({\bf a},b) \not = 0$ holds for at least one element $b\in D$.
It follows that 
$\sum_{b\in D}E^{\mathbf I}_{\psi}({\bf a },b)=1$, hence, we have that $$E_\varphi^{\mathbf I}({\bf a})=\prod_{b \in \mathrm{adom}(
{\mathbf I})} E_\psi^{\mathbf {I}}({\bf a},b).$$
Since $\psi^{\mathbf{A}({\mathbf I})}({\bf a}, b) \not = 0$ for some element $b\in D$, it follows that $({\bf a},b) \in \mathrm{adom}(\mathbf{I})^n$ (otherwise,  $\psi^{\mathbf{A}({\mathbf I})}({\bf a}, b) = 0$, by Definition \ref{brc-eval-def}). Consequently, ${\bf a} \in \mathrm{adom}(\mathbf{I})^{n-1}$, hence,  by Definitions \ref{FOL-sem-def} and \ref{brc-eval-def}, and the induction hypothesis, 
$$\varphi_E^{\mathbf{A}({\mathbf I})}({\bf a})=\prod_{b \in \mathrm{adom}(
{\mathbf I})} \psi^{{\mathbf A} (\mathbf {I})}({\bf a},b) = \prod_{b \in
\mathrm{adom}({\mathbf I})}
E_\psi^{\mathbf I}({\bf a},b)= E_\varphi^{\mathbf I}({\bf a}).$$
as desired. 
This completes the proof of Theorem \ref{Codd2}.
 \qedhere
\end{proof}

Immediate corollaries of Theorem~\ref{Codd2} include versions of Codd's Theorem for the bag semiring, the fuzzy semiring, the \L ukasiewicz semiring, and the tropical semiring on the natural numbers.

\section{Concluding Remarks}
The work presented here delineates the connections between relational algebra and relational calculus on databases over semirings by establishing two versions of Codd's Theorem in this setting. Along the way, new differences between the Boolean semiring and other semirings of importance to databases were unveiled. In particular, the five basic operations of relational algebra cannot express the division operation on the bag semiring. From the several directions for future research, we single out  two:
\begin{itemize}
\item For which semirings is it the case that basic relational algebra $\mathcal{BRA}$ can express the division operation? Is there a characterization of these semirings in terms of their structural properties?
\item In this paper, we had to make certain choices concerning the difference operation in relational algebra and the form of negation in relational calculus; to this effect, we chose the monus operation and the $\butnot$ (``but not'') connective, respectively. As described in~\cite{DBLP:conf/birthday/Suciu24, Bauer:2012}, there are other ways to define the difference operation in relational algebra over semirings. Are there versions of Codd's Theorem when alternative choices for the difference operation and for negation are made?
\end{itemize}
Finally, our results establish also connections between database theory and non-classical logics, since the $\butnot$ (``but not'') connective has featured in such logics. In particular, we saw earlier that relational calculus $\mathcal{RC}$, as defined here, has connections to the dual intuitionistic logic.  This paves the way for interaction between two communities, database theory and non-classical logics, that hardly had any interaction in the past.

\paragraph{Acknowledgments}
We  thank Dan Suciu for providing Example \ref{exam:suciu} of a naturally ordered semiring that cannot be expanded to a semiring with monus. We also thank Tomasz Kowalski for discussions about semirings with monus and for pointers to the literature.
%Note: Discuss what is the meaning of implication.
%\bibliographystyle{alpha}
\bibliography{BKN_article_ArXiV_v4}

\end{document}